%% file: main.tex
\documentclass{article}
\input{preamble}
\title{Online Learning in Stackelberg Security Games with Time-Varying Attack Intensities}
\author{%
  \begin{minipage}[t]{\dimexpr\textwidth-2\tabcolsep\relax}
    \centering
    Guanda Chen\textsuperscript{1}\quad
    Shiheng Zhang\textsuperscript{1}\quad
    Yue Wang\textsuperscript{2}\quad
    Yiding Ji\textsuperscript{1}\\[3pt]
    \normalfont\small
    \textsuperscript{1}Robotics and Autonomous Systems Thrust, Systems Hub,\\
    The Hong Kong University of Science and Technology (Guangzhou),
    Guangzhou, China\\[2pt]
    \textsuperscript{2}Department of Electrical and Computer Engineering,
    University of Central Florida, FL, USA\\[3pt]
    {\footnotesize\ttfamily
      gchen553@connect.hkust-gz.edu.cn, szhang593@connect.hkust-gz.edu.cn\\
      jiyiding@hkust-gz.edu.cn, yue.wang@ucf.edu}
  \end{minipage}%
}
\begin{document}
\maketitle

\input{sections/abstract}
\input{sections/introduction}
\input{sections/model}

\input{sections/oracle}
\input{sections/full_information}

\input{sections/bandit_main}
\input{sections/experiments_conclusion}
\label{page:main-end}
\input{sections/disclosures}
\bibliographystyle{iclr2027_conference}
\bibliography{references}
\clearpage
\appendix
\input{sections/simulation}
\input{sections/oracle_complexity_proof}

\input{sections/full_information_proof}
\input{sections/bandit_proofs}

\input{sections/resource_extension}
\input{sections/oracle_details}
\input{sections/oracle_runtime}
\clearpage
\input{sections/examples_experiments}
\end{document}

%% file: preamble.tex
\usepackage{arxiv_preprint,times}
\usepackage[utf8]{inputenc}
\usepackage[T1]{fontenc}
\usepackage{amsmath,amsfonts,amssymb,amsthm,bm}
\usepackage{graphicx,subcaption,booktabs,multirow,makecell,pifont}
\usepackage{microtype,nicefrac,enumitem,float,needspace,flafter}
\usepackage{hyperref}
\usepackage{url}
\hypersetup{hidelinks,
pdftitle={Online Learning in Stackelberg Security Games with Time-Varying Attack Intensities},
pdfauthor={Guanda Chen, Shiheng Zhang, Yue Wang, Yiding Ji},pdfsubject={Preprint},
pdfkeywords={Stackelberg security games, adversarial bandits, barycentric spanner, EXP2}}
\newtheorem{proposition}{Proposition}

\newtheorem{theorem}{Theorem}

\newtheorem{lemma}{Lemma}
\newtheorem{definition}{Definition}

\floatstyle{ruled}
\newfloat{algorithm}{tbp}{loa}
\floatname{algorithm}{Algorithm}
\numberwithin{equation}{section}
\allowdisplaybreaks[2]

%% file: sections/abstract.tex
\begin{abstract}
This work investigates no-regret online learning for Repeated Stackelberg Security Games (RSSGs) against attackers with time-varying attack intensities. Standard RSSG models often assume fixed attack intensities, limiting their ability to capture evolving attack capabilities. To bridge this gap, we formalize a model in which follower numbers, types, and attack intensities can vary across rounds.
We derive an exact mixed-integer linear programming oracle under optimistic tie-breaking. Under full-information feedback, combining this oracle with online learning algorithm yields expected $\mathcal{O}(\sqrt{T})$ regret against adversarial attacker sequences. Under bandit feedback, where only aggregate attack is observed, we represent the payoffs of all candidate defense strategies through a small set of basis strategies. This representation enables unbiased estimation of all candidate strategies' payoffs from a single aggregate payoff per round.
Building on this representation, we combine barycentric spanner with EXP2 to achieve expected $\mathcal{O}(\sqrt{T})$ regret against adversarial sequences. This guarantee contains no additional logarithmic factors in $T$ and matches an $\Omega(\sqrt{T})$ lower bound established within our model. Numerical simulations demonstrate the effectiveness of the proposed algorithms.
\end{abstract}

%% file: sections/introduction.tex
\section{Introduction}

In the Stackelberg Security Game formulation \citep{conitzer2006computing,bruckner2011stackelberg,li2017review,castiglioni2020online}, a leader (defender) first commits to a randomized patrolling or protection strategy, and a follower (attacker) observes and best responds to the leader's mixed strategy, by choosing a strategy to attack targets, in order to maximize its own payoff. This leader-follower paradigm has seen remarkable success in a wide range of domains, for instance,  it has been deployed for surveillance patrol scheduling at Los Angeles International Airport~\citep{pita2008deployed}, protection of critical infrastructures~\citep{shieh2012protect}, and wildlife conservation against illegal poaching~\citep{yang2014adaptive}. Beyond physical security, the same modeling principles have been extended to cyber-security challenges, including network resource allocation~\citep{bai2023stackelberg,gan2022defense} and signaling scheme design~\citep{castiglioni2020online,bernasconi2023optimal}.

However, existing studies and classical models typically assume that the attacker has a fixed attack capacity and a single target, which rarely holds in practice. In cyber attacks, for instance, the adversary can simultaneously attack or disable multiple servers, and this number of servers under attacks often grows over time as the adversary accumulates more compromised hosts \citep{antonakakis2017understanding}; Also, a network intruder who initially can only probe a single subnet may eventually launch a coordinated attack against multiple critical nodes \citep{ho2021hopper}. Analogously, in physical settings, poachers can maintain
many snares simultaneously, and the scale of their operations
may vary across seasons as hunting effort and resource
availability change \citep{macmillan2014factors,becker2013evaluating}. In these cases, the attacker's capacity to strike multiple targets varies dynamically, driven by exogenous shifts or adaptive resource allocation. Such time-varying attack intensities break the standard assumption of a fixed attack capacity, and, in turn, invalidate the assumption of fixed follower feasible action set that classical regret-minimization frameworks rely on. Motivated by these real-world phenomena, we aim to extend the existing models to Stackelberg Security Games with attackers exhibiting time-varying attack intensities, and to develop methodologies that allow the defender to adaptively learn and minimize regret against dynamically evolving attacks.


An additional challenge is to achieve optimal regret under
partial-information feedback, where the defender must learn
from aggregate attack outcomes without observing individual
follower types or attack intensities.
\citet{personnat2025learning} establish an $\Omega(\sqrt{T})$
regret lower bound in a related Bayesian multi-follower
Stackelberg setting.
Meanwhile, the adversarial bandit guarantees of
\citet{bernasconi2023optimal} and \citet{balcan2025nearly}
in their respective models scale as $O(\sqrt{T\log T})$
and $O(\sqrt{T}\log T)$ respectively.
These published bounds retain additional logarithmic factors
in the horizon.
This motivates the pursuit of an $O(\sqrt{T})$ regret
guarantee while accommodating multiple followers and
time-varying attack intensities.
We achieve this guarantee under aggregate bandit feedback
against adversarial follower sequences and establish a
matching $\Omega(\sqrt{T})$ lower bound within our model.
Consequently, our regret bound is minimax optimal in its
dependence on $T$, without additional horizon dependent
logarithmic factors.

\newcommand{\cmark}{\ding{51}}
\newcommand{\xmark}{\ding{55}}







\begin{table}[t]
\centering
\small
\setlength{\tabcolsep}{3pt}

\caption{Comparison of our results with related works under different
settings.}
\label{tab:algorithm_comparison}

\resizebox{\linewidth}{!}{%
\begin{tabular}{@{} l l ccccc @{}}
\toprule
\textbf{Results}
& \textbf{Feedback}
& \makecell{\textbf{Multiple}\\\textbf{followers}}
& \makecell{\textbf{Multiple}\\\textbf{types}}
& \textbf{Dynamic $l$}
& \textbf{Adversarial}
& \textbf{Regret} \\
\midrule

\multirow{2}{*}{\textbf{Our Algorithms}}
& Full Info    & \cmark & \cmark & \cmark & \cmark
& $\mathcal{O}(\sqrt{T})$ \\
& Partial Info & \cmark & \cmark & \cmark & \cmark
& $\mathcal{O}(\sqrt{T})$ \\
\midrule

\multirow{2}{*}{\citep{balcan2015commitment}}
& Full Info    & \xmark & \cmark & \xmark & \cmark
& $\mathcal{O}(\sqrt{T})$ \\
& Partial Info & \xmark & \cmark & \xmark & \cmark
& $\mathcal{O}(T^{2/3})$ \\
\midrule

\citep{bernasconi2023optimal}
& Partial Info & \cmark & \cmark & \xmark & \cmark
& $\mathcal{O}(\sqrt{T\log T})$ \\
\midrule

\multirow{2}{*}{\citep{harris2024regret}}
& Full Info    & \xmark & \cmark & \xmark & \cmark
& $\mathcal{O}(\sqrt{T\log T})$ \\
& Partial Info & \xmark & \cmark & \xmark & \cmark
& $\mathcal{O}(T^{2/3}(\log T)^{1/3})$ \\
\midrule

\multirow{2}{*}{\citep{personnat2025learning}}
& Full Info    & \cmark & \cmark & \xmark & \xmark
& $\mathcal{O}(\sqrt{T})$ \\
& Partial Info & \cmark & \cmark & \xmark & \xmark
& $\mathcal{O}(\sqrt{T\log T})$ \\
\midrule

\citep{balcan2025nearly}
& Partial Info & \xmark & \cmark & \xmark & \cmark
& $\mathcal{O}(\sqrt{T}\log T)$ \\

\bottomrule
\end{tabular}%
}
\end{table}

\textbf{Our Contributions.}

We generalize the classical SSG model by allowing the attacker's attack capacity to vary over time. The contributions of this work are summarized as follows:

 \emph{\textbf{Exact MILP formulation for intensity-constrained best responses.}} We model the follower's best-response dynamics by introducing auxiliary decision variables and constraints to encode the attacker's best-response selection under intensity-constrained action capacities. Based on this model, we develop an MILP via Big-M linearization. This yields an exact and tractable optimization oracle for computing optimal defender strategies under dynamically varying attack intensities.

\emph{\textbf{Online learning with adversarially evolving attackers.}} We then extend this model to a general online setting, where the number of attackers, types, and attack intensities all vary arbitrarily and adversarially across rounds. By integrating the MILP oracle with FPL, we obtain an efficient online algorithm that achieves sublinear regret despite the non-convex and discontinuous utility structure.

\emph{\textbf{Optimal learning from aggregate attack feedback.}}
When multiple followers act simultaneously, aggregate
observations generally do not reveal their individual types
or attack intensities.
We show that the defender can still achieve optimal
square-root regret by learning the payoffs of defense
strategies directly.
To do so, we express all candidate payoffs through a small
set of basis strategies.
This representation allows a single aggregate payoff
to provide unbiased estimates for every candidate strategy,
without recovering the underlying follower composition.
We combine these estimates with EXP2 and
barycentric spanner based exploration to handle adversarial
changes in follower numbers, types, and attack intensities.
Given game parameters, the resulting expected regret is
$O(\sqrt{T})$, with no additional logarithmic factors in $T$.
We also establish an $\Omega(\sqrt{T})$ lower bound within
the same model, showing that this dependence on the horizon
cannot be improved in general.
\subsection{Related Works}
\textbf{Online learning in Stackelberg games.} A growing number of work studies online learning in Stackelberg games, where early works \citep{letchford2009learning,bacchiocchi2024sample,blum2014learning,zhang2022no} focus on computing optimal defender strategies under repeated interactions. Some research has extended this framework to several dimensions; for instance, \citep{donahue2024impact} considers learning in cooperative Stackelberg settings, while \citep{harris2024regret,balcan2025nearly} incorporate side information and develop bandit learning algorithms.
In adversarial settings, \citep{balcan2015commitment,harris2024regret,balcan2025nearly} investigate regret minimization against arbitrary sequences of attacker types under both full and partial feedback. However, their models are limited to a single follower per round with a fixed, single-action attack capacity, and fail to accommodate settings with multiple attackers or time-varying attack intensities.
More recent works like \citep{personnat2025learning,cheng2022single} consider settings with multiple simultaneous followers. Nevertheless, they assume stochastic or distributional models over attackers, rather than considering the fully adversarial attacking sequences, and do not capture time-varying attack intensities.

\textbf{Modeling diverse attack behaviors.} To better capture attacker behaviors, several studies consider richer attacker capabilities through Bayesian models and application specific formulations.
For example, \citet{wahab2019resource} study heterogeneous
attack types in a Bayesian framework, while
\citet{yu2025optimal,xu2026stackelberg} investigate
resource-constrained attacks in control and cyber-physical
systems.
Application specific models, such as DDoS defense in IoT networks \citep{chen2021ddos}, further approximate attack intensities via discretization or known distributions. Despite these advances, existing approaches remain limited in generality. Most works rely on strong structural assumptions, such as stochastic attack models or application-specific dynamics, which do not consider dynamic and adversarial variations in attacker capabilities.

\textbf{Regret guarantees under bandit feedback.}
The adversarial bandit
methods of \citet{balcan2015commitment,harris2024regret} conduct exploration rounds within each
time block and achieve $T^{2/3}$ dependence.
Balancing the repeated exploration cost against regret
from block-level updates yields this dependence on $T$.
For contextual Stackelberg games,
\citet{balcan2025nearly} obtain
$O(\sqrt{T}\log T)$ regret through logdet-FTRL,
whose regularization and
covariance-estimation analysis retain logarithmic factors. In addition, a direct instantiation of the logdet-FTRL
approach of \citet{balcan2025nearly} in our setting
can be infeasible when structural dependencies make
all feasible strategy matrices singular, leaving
the log-determinant regularizer without a finite value.
For multiple simultaneous followers,
\citet{personnat2025learning} obtain
$O(\sqrt{T\log T})$ regret under partial information
and establish a $\Omega(\sqrt{T})$ lower bound.
In their UCB-based method, the additional
$\sqrt{\log T}$ factor arises from confidence bounds
used to control utility estimation errors. Within our model, we achieve
$O(\sqrt{T})$ expected regret under bandit feedback
against adversarial follower sequences, while allowing
multiple simultaneous followers and time-varying
attack intensities.
This bound matches the $\Omega(\sqrt{T})$ lower bound
established within our model
(Proposition~\ref{prop:horizon-lower-bound})
and is therefore optimal in its dependence on $T$.

Table~\ref{tab:algorithm_comparison} summarizes these
regret guarantees alongside the corresponding modeling
assumptions and feedback settings, where $l$ denotes
the attack intensity.



%% file: sections/model.tex
\section{RSSGs with Time-Varying Attack Intensity}
\label{Problem}

\paragraph{SSG Modeling of Time-varying Attack Intensity.}
Under the time-varying attack intensity, the problem is formulated as an extended Repeated Stackelberg Security Game (RSSG). The interaction between the defender and the attacker is characterized as follows: (1) each edge corresponds to a target that may be attacked by the attacker, and (2) the weight allocation vector represents the defender's mixed strategy over targets. Based on this formulation, the local decision-making process of defender is modeled as an RSSG, denoted by $\mathcal{G}$.

\begin{definition}\label{resilientGame}
    An RSSG is defined as a tuple $\mathcal{G} = (T, \mathcal{E}^-, {\cal E}_k, \mathcal{W}, {\cal{U}})$, where 1) $T$ is the time horizon; 2) ${\mathcal{E}}^-=\{1,...,N\}$ is set of targets with cardinality $N$; 3) ${\cal E}_k \subseteq {\mathcal{E}}^-$ denotes the subset of attacked targets selected by the attacker of type $k$; 4) ${\mathcal{W}} = \{ \mathbf{w} \in \mathbb{R}_{\geq 0}^{N} \mid \sum_{{j}} w_{j} = 1 \}$ is the defender's mixed strategy space where $\mathbf{w}$ is a weight allocation vector; and 5) ${\cal{U}}$ specifies the utility structure. 
\end{definition}

The utility structure is defined as follows. If target $j$ is attacked by attacker of type $k\in\{1,\dots,K\}$ and not covered by the defender, the defender and the attacker receive utilities $\lambda_{j}^{r} \in [-1,0]$ and $\lambda_{j}^{k} \in [0,1]$, respectively. If target $j$ is attacked and covered by the defender, the defender and the attacker receive utilities $\rho_{j}^{r} \in [0,1]$ and $\rho_{j}^{k} \in [-1,0]$, respectively. Given that the defender protects target $j$ with probability $w_{j}$, the expected utility of the type $k$ attacker when attacking target $j$ is ${u}_a^k({\cal E}_{j}, \mathbf{w})= \lambda_{j}^k (1-w_{j}) + \rho_{j}^k w_{j}$.
The attacker's maximum capacity under fully available resources is $F$, i.e., the attacker can choose to attack at most $F$ targets at each time step. The attack intensity $l$, defined as the number of targets the attacker chooses to attack, is determined at the beginning of each round. Given the defender's mixed strategy $\mathbf{w}$, the attacker selects $l$ targets to attack. The attacker's best response is defined as the subset of targets that maximizes its total utility: $\mathbf{b}_k(\mathbf{w}) = \arg \max_{{\cal{E}}_k \cap \mathcal{E}^-} \sum_{\mathcal{E}_{j} \in {{\cal{E}}_k \cap \mathcal{E}^-}} u_a^k({\cal E}_{j}, \mathbf{w}), \text{s.t.} \ |{\cal{E}}_k \cap \mathcal{E}^-|  = l$.

Throughout the paper, we use an optimistic tie-breaking rule and any remaining ties are broken by increasing target index. The defender and the attackers know each other's utility parameters. At the end of the round, full-information feedback reveals the realized follower types and intensities, whereas bandit feedback reveals only the aggregate attack outcome specified in Section~\ref{Partial}.

%% file: sections/oracle.tex
\section{Exact Computation for Extended RSSG}
\label{Single}

RSSG with time‑varying attack intensity poses a new computational challenge: classic solvers like DOBSS~\cite{paruchuri2008playing} are no longer suitable because the follower action set changes with the realized intensity. In this section, we design an exact algorithm for the defender such that, given a prior distribution ${\mathbf{q}}=[q_1,\dots,q_F]$ over attack intensities (e.g., the empirical frequencies), it computes the optimal strategy, where $q_l$ represents the defender’s prior probability for attacker of attack intensity $l \in \mathcal{L} = \{1, \dots, F\}$. The current intensity is fixed before the defender makes its
decision and is not observable at the time of decision making. We first formulate the problem as a bi-level optimization
problem and then reduce it to an equivalent single-level MILP.

\paragraph{Bi-Level Formulation.}
We first consider a fixed attacker type $k$. For each attack intensity $l$, we introduce a binary decision
vector to represent the attacker's best response,
$\mathbf{h}_l^k=[h_{1,l}^k,\dots,h_{N,l}^k]^\top
\in\{0,1\}^{N}$,
where $h_{j,l}^k=1$ indicates that target $\mathcal{E}_j$
is selected by the attacker, and $h_{j,l}^k=0$ otherwise.
Let $f_l^k(\mathbf{w})$ denote the defender's expected utility
against an attacker of type $k$ and intensity $l$ under the
tie-breaking rule in Section~\ref{Problem}, written as: 
\begin{equation}
\label{subobj:milp}
f_l^k(\mathbf{w})
=
\sum_{\mathcal{E}_j\in\mathcal{E}^-}
h_{j,l}^k
\left[
\lambda_j^r(1-w_j)+\rho_j^r w_j
\right].
\end{equation}

Given a prior attack intensity distribution $\mathbf{q}$,
the defender aims to maximize its expected utility
$\sum_{l=1}^{F}q_l f_l^k(\mathbf{w})$
by selecting an optimal strategy $\mathbf{w}$.
This leads to the following bi-level formulation:
\begin{subequations}
\label{eq:bilevel_milp}
\begin{align}
\max_{\mathbf{w}\in\mathcal{W},\,\mathbf{h}_{1:F}^k}
\quad&
\sum_{l=1}^{F}q_l
\sum_{\mathcal{E}_j\in\mathcal{E}^-}
h_{j,l}^k
\left[
\lambda_j^r(1-w_j)+\rho_j^r w_j
\right]
\label{obj:milp}
\\
\text{s.t.}\quad&
\mathbf{h}_l^k
\in
\arg\max_{\substack{
\mathbf{h}'\in\{0,1\}^{N}\\
\sum_j h'_j=l}}
\sum_{\mathcal{E}_j\in\mathcal{E}^-}
h'_j u_a^k(\mathcal{E}_j,\mathbf{w}),
\quad l=1,\dots,F.
\label{attacker_const}
\end{align}
\end{subequations}

In this formulation, \eqref{obj:milp} represents the defender's
upper-level objective.
The constraint $\mathbf{w}\in\mathcal{W}$ ensures that the
strategy vector remains within the probability simplex.
The lower-level constraint~\eqref{attacker_const} captures
the attacker's best response.
Given a fixed strategy $\mathbf{w}$ and attack intensity $l$,
the attacker selects $l$ targets with the highest utilities.
By introducing the response variables
$\{\mathbf{h}_l^k\}_{l=1}^{F}$, the formulation unifies all
possible attack intensities within a single optimization
framework.

\paragraph{Single-Level MILP Reduction}
To enable efficient computation, we transform the bi-level problem into an equivalent single-level MILP. This is achieved by 1) linearizing the bilinear terms, and 2) encoding the attacker’s best-response conditions using Big-M constraints, which is formulated as
\begin{subequations}
\label{eq:consensus_milp_noself}
\begin{align}
&\max_{\mathbf{w},\,\mathbf{z}_{1:F}^k,\,
       \mathbf{h}_{1:F}^k,\,\tau_{1:F}^k}
\quad
\sum_{l=1}^{F}q_l
\sum_{\mathcal{E}_j\in\mathcal{E}^-}
\left[
\lambda_j^r h_{j,l}^k
+
(\rho_j^r-\lambda_j^r)z_{j,l}^k
\right]
\label{obj}
\\
\text{s.t.}\quad&
\mathbf{w}\in\mathcal{W},
\label{cons1}
\\
&
z_{j,l}^k\leq w_j,
\quad
z_{j,l}^k\geq w_j+h_{j,l}^k-1,
\label{con:z_lb}
\\
&
\sum_{\mathcal{E}_j\in\mathcal{E}^-}h_{j,l}^k=l,
\quad
\mathbf{h}_l^k\geq\mathbf{h}_{l-1}^k
\quad(l\geq2),
\label{const:nested_logic}
\\
&
u_a^k(\mathcal{E}_j,\mathbf{w})
\geq\tau_l^k-B(1-h_{j,l}^k),
\quad
u_a^k(\mathcal{E}_j,\mathbf{w})
\leq\tau_l^k+B h_{j,l}^k,
\label{const:bigM_upper}
\\
&
z_{j,l}^k\in[0,1],
\quad
h_{j,l}^k\in\{0,1\},
\quad
z_{j,l}^k\leq h_{j,l}^k,
\quad
-1\leq\tau_l^k\leq1.
\label{const:z}
\end{align}
\end{subequations}
The components and constraints of
\eqref{eq:consensus_milp_noself} are detailed as follows:
\begin{itemize}[leftmargin=*]
    \item \textit{Linearized Payoff Structure
    \eqref{con:z_lb} and \eqref{const:z}:}
    To address the nonlinearity introduced by the bilinear
    term $h_{j,l}^k w_j$, we introduce an auxiliary variable
    $z_{j,l}^k$ to represent their product.
    This enables a standard linearization via McCormick-type
    constraints, ensuring $z_{j,l}^k=h_{j,l}^k w_j$
    at optimality.
    As a result, the payoff can be expressed as a linear
    combination of the reward term $\rho_j^r$ and the
    penalty term $\lambda_j^r$.

    \item \textit{Attacker's Cardinality and Nested Logic
    \eqref{const:nested_logic}:}
    The cardinality constraint selects exactly $l$ targets,
    while the nesting constraints impose a nested structure
    across attack intensities.
    For a fixed defender strategy and attacker type $k$,
    higher attack intensities expand upon previously
    selected targets.

    \item \textit{Optimality Conditions via Big-M
    \eqref{const:bigM_upper}--\eqref{const:z}:}
    These conditions represent the optimality conditions
    for the attacker's best response~\eqref{attacker_const}.
    For each $l$, a scalar variable $\tau_l^k$ is introduced
    to represent a utility threshold: any target with
    attacker utility
    $\lambda_j^k(1-w_j)+\rho_j^k w_j>\tau_l^k$
    is forced to be selected ($h_{j,l}^k=1$), while targets
    with utility below the threshold are set to
    $h_{j,l}^k=0$.
    Equality at the threshold is permitted, allowing
    tie-breaking in favor of defender.
    Here, $B$ is a sufficiently large positive constant.
\end{itemize}

%% file: sections/full_information.tex
\section{Online Learning with Full Information Feedback}
\label{MultipleF}

Based on the above optimization framework, we consider a more general and challenging setting in which the attacker sequence evolves arbitrarily over time. In this case, the number of attackers, their types, and their attack intensities vary over time, even potentially in an adversarial manner. Existing works only partially address this setting. Prior studies on adversarial Stackelberg learning and multiple-attacker models (e.g., \cite{balcan2015commitment,harris2024regret,personnat2025learning}) do not consider time-varying attack intensities, while works incorporating attack intensities (e.g., \cite{qian2016restless,xu2026stackelberg}) rely on structured dynamics and cannot handle adversarially chosen attacker sequences or multiple followers. This gap motivates us to study a unified setting that simultaneously captures
(i) multiple attackers with diverse attacker types,
(ii) time-varying attack intensities, and
(iii) adversarially evolving attacker sequences.



\paragraph{Attackers with Arbitrary Types and Intensities.}
We consider a setting where at each time step $t$,
the number of followers $n(t)\le C$, their types,
and their attack intensities can be chosen arbitrarily.
The follower configuration is non-anticipating:
it may depend on past interactions but is fixed before
the defender uses its current private randomness to
choose $\mathbf{w}(t)$.
This setting captures practical situations in which
attacker characteristics, including their number, types,
and attack intensities, change over time due to resource
constraints, coordination effects, and environmental
uncertainty.

Under full-information feedback, the current follower
configuration remains hidden while the defender chooses
$\mathbf{w}(t)$.
After observing this strategy, the followers select their
attack targets according to the prescribed best responses.
The configuration is revealed only at the end of the round,
when the defender observes every follower's type and
attack intensity and updates the empirical count matrix.



Let $\mathbf{M} \in \mathbb{R}^{K \times F}$ be the matrix with entrie $(\mathbf{M})_{k,l} = m_{kl}$, where $m_{kl}$ denotes the number of followers of type $k$ with intensity $l$.
Similarly, define $\mathbf{F}(\mathbf{w}) \in \mathbb{R}^{K \times F}$ with entries $(\mathbf{F}(\mathbf{w}))_{k,l} = f_l^k(\mathbf{w})$. The total utility can then be compactly expressed as the sum of element-wise products
$\sum_{k=1}^{K}\sum_{l=1}^{F} \bigl(\mathbf{M} \odot \mathbf{F}(\mathbf{w})\bigr)_{kl}$,
where $\odot$ denotes the Hadamard (entrywise) product. 

At round $t$, let $\mathbf{M}(t)$ denote the realized
follower count matrix and define its empirical average as
$\hat{\mathbf{M}}(t)=\frac{1}{t}
\sum_{t'=1}^{t}\mathbf{M}(t')$.
Their entries are denoted by $m_{kl}(t)$ and
$\hat m_{kl}(t)$, respectively.
Empirical average utility of strategy $\mathbf{w}$
is obtained by replacing $\mathbf{M}$ with
$\hat{\mathbf{M}}(t)$ in the above utility expression.




\textbf{Exact Computation for Attacker Sequences.}
Given the empirical count matrix $\hat{\mathbf{M}}(t)$,
we extend the single-level MILP in
\eqref{eq:consensus_milp_noself} to multiple attacker types.
Specifically, we impose the linearization and best-response
constraints for each type while retaining a common defender
strategy $\mathbf{w}$. We weight each payoff
$f_l^k(\mathbf{w})$ by empirical
count $\hat m_{kl}(t)$.
The resulting oracle computes an exact maximizer of the
defender's empirical average utility.
The complete MILP formulation is provided in
Appendix~\ref{app:oracle-details}.

\begin{proposition}[Oracle size and enumeration bound]
\label{Milp_algorithm_complexity}
The MILP oracle~\ref{app:oracle-details} has
$O(KFN)$ variables and constraints.
It admits an exact enumeration-based implementation
that solves at most
$\bigl(N!/(N-F)!\bigr)^K=O(N^{KF})$
linear programs, each with $O(KFN)$ variables and constraints.
\end{proposition}

This bound applies to the enumeration-based LP implementation.
The compact MILP formulation allows direct solution using
standard MILP solvers such as Gurobi~\citep{gurobi}.
These solvers can exploit presolve, cutting planes,
and branch-and-bound to prune the search space without
requiring explicit enumeration of all follower response
profiles, yielding substantial computational
savings in practice. Our runtime experiments in Appendix~\ref{app:oracle-runtime} further support this observation, showing that the compact MILP formulation achieves substantially shorter solution times than LP enumeration on the larger tested instances.

However, deterministic strategy selection based on past
empirical counts can make the defender's next allocation
predictable, and adversaries may exploit this predictability by choosing
unfavorable follower configurations before the defender acts.
As shown in Appendices~\ref{Appendix11},
\ref{Appendix12}, and~\ref{Appendix_Simulation22},
a carefully constructed adversarial sequence can repeatedly
mislead the defender, potentially leading to linear regret.
This motivates combining the exact oracle with randomized
perturbations through the FPL framework.

\paragraph{No-Regret Learning Framework with Perturbations.} To mitigate this vulnerability, we adopt an online decision-making approach based on the FPL framework \citep{kalai2005efficient}. The main idea is to introduce stochastic perturbations to the empirical count matrix $\hat{\mathbf{M}}(t)$, thereby smoothing the leader's decision boundary and preventing adversary from predicting defender's actions. Concretely, at each round $t$, we perturb the matrix as $\tilde{\mathbf{M}}(t) = \frac{t-1}{t}\hat{\mathbf{M}}(t-1) + \boldsymbol{\epsilon}(t),$ where $\boldsymbol{\epsilon}(t) \sim \text{Uniform}[0, 1/(\delta\sqrt{t})]^{K \times F}$. Then we solve the resulting perturbed optimization problem using our proposed MILP formulation~\ref{app:oracle-details} as an oracle. This integration of online perturbation and MILP-based
optimization achieves sublinear expected regret against
arbitrarily evolving, non-anticipating attacker sequences. The details are summarized in Algorithm~\ref{alg:fpl-VaryingN}.

\begin{algorithm}[tbp]
\caption{FPL with an MILP oracle for time-varying-intensity RSSGs}
\label{alg:fpl-VaryingN}
\small
\textbf{Input:} $T,F,K,\delta$ and utility parameters.
Initialize $\hat{\mathbf M}(0)=\mathbf0$.
\begin{enumerate}[leftmargin=1.6em]
\item For $t=1,\ldots,T$:
\begin{enumerate}[label=(\alph*),leftmargin=1.5em]
\item A non-anticipating environment fixes the hidden $\mathbf M(t)$ using the past.
For $t>1$, set $\hat{\mathbf M}(t-1)=(t-1)^{-1}\sum_{s<t}\mathbf M(s)$.
\item Sample $\boldsymbol\epsilon(t)\sim\mathcal U[0,1/(\delta\sqrt t)]^{K\times F}$ and set
$\tilde{\mathbf M}(t)=\frac{t-1}{t}\hat{\mathbf M}(t-1)+\boldsymbol\epsilon(t)$.
\item Solve the MILP~\eqref{fpl_TimeVarying n} using $\tilde{\mathbf M}(t)$
and play the maximizing strategy $\mathbf w(t)$.
\item Followers observe $\mathbf w(t)$ and play their prescribed best responses.
Observe the full configuration $\mathbf M(t)$ after the round.
\end{enumerate}
\end{enumerate}
\end{algorithm}


Define the \emph{cumulative regret} as the difference between the cumulative utility of the best fixed strategy in hindsight and that obtained by the algorithm, written as: $R(T)=\max_{\mathbf{w} \in {\mathcal{W}}} \sum_{t=1}^T( \sum_{k=1}^{K}\sum_{l=1}^{F} \bigl({\mathbf{M}}(t) \odot \mathbf{F}(\mathbf{w})\bigr)_{kl}) - \sum_{t=1}^T(\sum_{k=1}^{K}\sum_{l=1}^{F} \bigl({\mathbf{M}}(t) \odot \mathbf{F}(\mathbf{w}(t))\bigr)_{kl})$


\begin{theorem}[Regret bound for Algorithm~\ref{alg:fpl-VaryingN}]
\label{thm2:fpl_regret}
Let $\{\mathbf{w}(t)\}_{t=1}^{T}$ be the sequence of
strategies generated by Algorithm~\ref{alg:fpl-VaryingN}
with parameter $\delta>0$ and perturbations
$\boldsymbol{\epsilon}(t)
\sim \mathcal{U}[0,1/(\delta\sqrt{t})]^{K\times F}$.
For any non-anticipating sequence of follower count
matrices $\{\mathbf{M}(t)\}_{t=1}^{T}$ with $n(t)\le C$,
the expected cumulative regret satisfies
\(
    \mathbb{E}[R(T)]
    \le
    \left(
        4\delta FC^2
        + \frac{KF(F+1)}{\delta}
    \right)\sqrt{T}.
\)
In particular, choosing
$\delta=\frac{\sqrt{K(F+1)}}{2C}$ yields
    $\mathbb{E}[R(T)]
    \le 4CF\sqrt{K(F+1)T}.$
\end{theorem}

Moreover, oracle \ref{app:oracle-details} is not limited to unit-sum allocation ($\sum w_j=1$) but can naturally generalize to scenarios where the leader can distribute defense budget greater than one, simply by modifying the simplex constraint accordingly. This flexibility allows the proposed framework to handle a wide range of security allocation problems with multiple attacker types and multi-target attacks efficiently.

%% file: sections/bandit_main.tex
\section{Learning under partial information}
\label{sec:partial-information}\label{Partial}\label{Multiple}


In our partial information back setting, we consider a setting with at most $C$ attackers acting
simultaneously in each round. The follower sequence is arbitrary
but fixed before play begins. Under bandit feedback, the defender
observes only the aggregate number of attacks on each target,
while the number, types, and attack intensities of
followers remain unobserved.
The methods of
\citet{balcan2015commitment,harris2024regret,balcan2025nearly}
study single-follower interactions, whereas
\citet{personnat2025learning} study multiple followers under
stochastic type distributions.
To our knowledge, prior work has not explicitly addressed
bandit learning with simultaneous multiple followers,
arbitrarily varying attack intensities, and aggregate feedback.

We therefore aim to develop a no-regret algorithm for
the above adversarial multi-follower setting using only
aggregate feedback.
Specifically, our goal is to achieve
$\mathcal{O}(\sqrt{T})$ expected regret for fixed game
parameters, without the additional horizon dependence
in the guarantees.

\subsection{Bandit model and finite candidate strategies}
\label{sec:bsxp-model}

Let $h_{j,l}^k(\mathbf{w})$ denote the best-response
indicator defined in Section~\ref{Single}.
After committing to $\mathbf{w}(t)$, the defender observes
the aggregate attack counts $\nu_j(t)
    =
    \sum_{k=1}^{K}\sum_{l=1}^{F}
    m_{kl}(t)h_{j,l}^k(\mathbf{w}(t)),$ where $ j=1,\dots,N.$
The round-$t$ payoff can be written as $
    r_t(\mathbf{w}(t))= \sum_{k=1}^{K}\sum_{l=1}^{F}
    m_{kl}(t)f_l^k(\mathbf{w}(t))
    =
    \sum_{j=1}^{N}\nu_j(t)
    \left[
        \lambda_j^r+
        (\rho_j^r-\lambda_j^r)w_j(t)
    \right].$
In each round, at most $C$ followers attack up to $CF$
targets.
Since the defender's utility from each attack lies
in $[-1,1]$, the total payoff satisfies
$|r_t(\mathbf{w})|\le CF$.
We use the regret definition in Section~\ref{MultipleF},
with expectations taken over the learner's randomization.

\paragraph{Structured Decision Set.}
For a fixed follower response profile, the associated
defender payoff is linear in $\mathbf{w}$.
This motivates constructing a finite candidate set from
the vertices of best-response polytopes.
For a profile $\sigma=(\sigma_l^k)_{k,l}$ with
$\sigma_l^k\in\{0,1\}^{N}$ and
$\sum_j\sigma_{j,l}^k=l$, define
\begin{equation}
P(\sigma)=\bigl\{\mathbf{w}\in\mathcal{W}: u_a^k(\mathcal{E}_i,\mathbf{w})\ge u_a^k(\mathcal{E}_j,\mathbf{w}) \text{ whenever }\sigma_{i,l}^k=1,\ \sigma_{j,l}^k=0,\quad \text{for all }k,l\bigr\}.
\label{eq:bsxp-regions}
\end{equation}
These polytopes cover $\mathcal{W}$.
Let
\begin{equation}
    V=
    \bigcup_{\sigma:\,P(\sigma)\ne\varnothing}
    \operatorname{vert}(P(\sigma)).
    \label{eq:bsxp-V}
\end{equation}
The following lemma shows that, under the prescribed
optimistic tie-breaking rule, restricting the defender
to $V$ preserves the best fixed payoff exactly. To see why $V$ suffices, take any $\mathbf{w}\in\mathcal{W}$
and let $\sigma$ be its follower response profile.
Keeping $\sigma$ fixed, the cumulative defender payoff
is linear and attains a maximum at some vertex of
$P(\sigma)$.
At that vertex, $\sigma$ remains a best response, although
other profiles may also become optimal.
Optimistic tie-breaking rule chooses $\sigma'$ with highest defender's expected utility as the attacker's best response.
Thus, the best fixed strategy can be chosen from $V$
without any loss in cumulative payoff, establishing
the exact reduction below.

\begin{lemma}[Exact strategy reduction]
\label{lem:bsxp-finite}
For follower sequence under partial information feedback,
\begin{equation}
    \max_{\mathbf{w}\in\mathcal{W}}
    \sum_{t=1}^{T}r_t(\mathbf{w})
    =
    \max_{\mathbf{w}\in V}
    \sum_{t=1}^{T}r_t(\mathbf{w}).
    \label{eq:bsxp-finite-equality}
\end{equation}
Moreover, $V$ is independent of $T$ and satisfies $|V|
    \le
    \binom{N+K\binom{N}{2}}{N-1}.$
\end{lemma}

The proof is given in
Appendix~\ref{app:bsxp-finite-proof}.

The lemma reduces continuous set $\mathcal W$ to finite set $V$, but $V$ may
still be large.
For each $\mathbf w\in V$, the matrix $F(\mathbf w)$
records its payoff against every attacker type and
attack intensity.
The followers present in a round determine how these
payoffs are combined, while $F(\mathbf w)$ remains fixed. The payoff matrices of all candidates span a space of
dimension at most $KF$.
We therefore choose basis strategies from $V$ whose
payoff matrices span this space. Each candidate's payoff matrix is a known linear
combination of the basis payoff matrices.
Since the round payoff is linear in these payoff matrices,
the same coefficients express its payoff in terms
of the basis strategies' payoffs.
By incorporating randomized exploration over the basis
strategies, we construct unbiased estimates of their
payoffs from the single aggregate payoff observed
in each round.
Combining these estimates with the known coefficients
then yields unbiased payoff estimates for all candidates.



\subsection{Payoff coordinates and the EXP2 algorithm}
\label{sec:bsxp-algorithm}

Define
$
    \mathbf a(\mathbf{w})=\operatorname{vec}(F(\mathbf{w}))\in\mathbb R^{KF}
    $, where $
    d=\dim\operatorname{span}\{\mathbf a(\mathbf{w}):\mathbf{w}\in V\}\le KF
    \label{eq:bsxp-features}
$.
Apply the same vectorization for $M(t)$. 
By the barycentric-spanner existence result of
\citet{awerbuch2008online}, the finite payoff-feature set admits a
basis $ B=\{\mathbf a(\mathbf{w}_1),\ldots,\mathbf a(\mathbf{w}_d)\}$
such that every feature has coordinates in $[-1,1]^d$. Put
\begin{equation}
    B=[\,\mathbf a(\mathbf{w}_1)\ \cdots\ \mathbf a(\mathbf{w}_d)\,],
    \qquad
    \mathbf c(\mathbf{w})=(B^\top B)^{-1}B^\top\mathbf a(\mathbf{w}).
    \label{eq:bsxp-basis}
\end{equation}
Since the columns of $B$ form a basis, these coordinates are unique.
Spanner property therefore gives,
\begin{equation}
    \mathbf a(\mathbf{w})=B\mathbf c(\mathbf{w}),\qquad
    \|\mathbf c(\mathbf{w})\|_\infty\le1,\qquad
    \mathbf c(\mathbf{w}_i)=\mathbf e_i.
    \label{eq:bsxp-spanner-properties}
\end{equation}
Each basis element
is the payoff feature of a complete strategy $\mathbf{w}_i$, so it can be
sampled by playing that strategy. Define the normalized loss and the unobserved vector of basis losses by 
\begin{equation}
    \ell_t(\mathbf{w})=-\frac{r_t(\mathbf{w})}{CF},
    \qquad
    \mathbf g_t=(\ell_t(\mathbf{w}_1),\ldots,\ell_t(\mathbf{w}_d))^\top
    =-\frac{B^\top\operatorname{vec}(M(t))}{CF}.
    \label{eq:bsxp-normalization}
\end{equation}
The loss vector of strategy $\mathbf{w}$ can be written as: $\ell_t(\mathbf{w})=\mathbf c(\mathbf{w})^\top\mathbf g_t$, where $|\ell_t(\mathbf{w})|\le1.$
This equation provides the second
reduction. Although best responses make the payoff map
discontinuous in $\mathbf{w}$, the losses of all strategies
in $V$ depend on the unobserved follower sequence through
only $d\le KF$ basis losses. The payoff from any
played $\mathbf{w}\in V$ is one linear observation of these
losses, with known coefficients $\mathbf c(\mathbf{w})$. Let $\mu$ be uniform over $\mathbf{w}_1,\ldots,\mathbf{w}_d$,
and set $p_t=(1-\gamma)x_t+\gamma\mu$, where $x_t$ is
the exponential-weights distribution.
Since $\mathbf c(\mathbf{w}_i)=\mathbf e_i$, the
second-moment matrix $S_t$ is written as:
\begin{equation}
    S_t=
    \sum_{\mathbf{w}\in V}
    p_t(\mathbf{w})\mathbf c(\mathbf{w})
    \mathbf c(\mathbf{w})^\top
    =(1-\gamma)\sum_{\mathbf{w}\in V}
          x_t(\mathbf{w})\mathbf c(\mathbf{w})\mathbf c(\mathbf{w})^\top
          +\frac{\gamma}{d}I_d.
    \label{eq:bsxp-second-moment}
\end{equation}
Thus, $S_t$ is invertible and computable from the known
coordinates and $p_t$.
After sampling $\mathbf{w}(t)\sim p_t$ and observing
its payoff, we have estimators: $\widehat{\mathbf g}_t
    =S_t^{-1}\mathbf c(\mathbf{w}(t))
    \ell_t(\mathbf{w}(t))
    $ and $
    \widehat\ell_t(\mathbf{w})
    =\mathbf c(\mathbf{w})^\top\widehat{\mathbf g}_t.$
Every played strategy, including one sampled through the exploitation
part of $p_t$, contributes to this estimator.
Because the basis strategies can actually be played, mixing
with $\mu$ gives $S_t\succeq(\gamma/d)I_d$. Thus
$\mathbb E_t[\widehat{\mathbf g}_t]=\mathbf g_t$, and
$\ell_t(\mathbf{w})=\mathbf c(\mathbf{w})^\top\mathbf g_t$ gives unbiased loss estimates for
every strategy in $V$. The estimator uses the observation from
every round without identifying individual followers.

Algorithm~\ref{alg:bsxp} instantiates the EXP2 algorithm \citep{bubeck2012towards} in payoff-basis coordinates.
All estimates are computable from aggregate
feedback and defender's utility matrix.

\subsection{Regret guarantee}
\label{sec:bsxp-guarantee}

\begin{theorem}[Square-root regret under aggregate feedback]
\label{thm:bsxp-main}
Consider the model of Section~\ref{sec:bsxp-model}, with $V$ defined
by \eqref{eq:bsxp-V} and a basis satisfying
\eqref{eq:bsxp-spanner-properties}. Suppose $d\ge1$ and $|V|\ge2$.
For any $\eta>0$, $0<\gamma<1$ with $\eta d^2\le\gamma$,
Algorithm~\ref{alg:bsxp} satisfies $\mathbb E[R(T)]
    \le CF\left(
        \frac{\log|V|}{\eta}
        +\frac{\eta dT}{1-\gamma}
        +2\gamma T
    \right).$
In particular, set
\begin{equation}
    \eta=\min\left\{
       \frac{1}{2d^2},
       \frac{1}{2d}\sqrt{\frac{\log|V|}{T}}
    \right\},
    \qquad
    \gamma=\eta d^2.
    \label{eq:bsxp-parameters}
\end{equation}
Then, for every $T\ge1$, 
\begin{equation}
    \mathbb E[R(T)]
       \le 4CFd\sqrt{T\log|V|}.
    \label{eq:bsxp-main-bound}
\end{equation}
\end{theorem}
The proof is given in Appendix~\ref{app:partial-proofs}:
Appendix~\ref{app:bsxp-estimation} proves conditional unbiasedness and
the relevant second-moment bounds, and
Appendix~\ref{app:bsxp-ew} proves the exponential-weights inequality.

\begin{algorithm}[tbp]
\caption{Barycentric-spanner EXP2 with aggregate attack feedback}
\label{alg:bsxp}
\small
\textbf{Input:} Utility parameters;
$N,K,F,C,T$; $\eta>0$ and $\gamma\in(0,1)$.
\begin{enumerate}[leftmargin=1.7em,label=\arabic*.,nosep]
    \item Construct $V$ by \eqref{eq:bsxp-V}, compute its
    payoff features, and determine $d$.
    If $d=0$ or $|V|=1$, play any $\mathbf{w}\in V$
    in every round and stop.

    \item Compute $B$ and $\mathbf c(\mathbf{w})$
    satisfying \eqref{eq:bsxp-spanner-properties}.
    Let $\mu$ be uniform over the basis strategies,
    and initialize $x_1(\mathbf{w})=1/|V|$
    for all $\mathbf{w}\in V$.

    \item \textbf{For} $t=1,\ldots,T$:
    \begin{enumerate}[leftmargin=1.6em,label=(\alph*),nosep]
        \item Set $p_t=(1-\gamma)x_t+\gamma\mu$ and
        compute $S_t$ by \eqref{eq:bsxp-second-moment}.

        \item Sample $\mathbf{w}(t)\sim p_t$ and commit
        to it. Followers then best respond.

        \item Observe $\nu(t)$, compute
        $\ell_t(\mathbf{w}(t))=-r_t(\mathbf{w}(t))/(CF)$,
        and estimate all candidate losses.

        \item Update
        $x_{t+1}(\mathbf{w})\propto
        x_t(\mathbf{w})\exp(-\eta\widehat\ell_t(\mathbf{w}))$
        for all $\mathbf{w}\in V$, normalized to sum to one.
    \end{enumerate}
\end{enumerate}
\end{algorithm}

The aggregate payoff observed in each round is used to estimate
the losses of all candidate strategies. Crucially, the weighted
conditional second moment satisfies
$\mathbb E_t[\sum_{\mathbf{w}\in V}
x_t(\mathbf{w})\widehat\ell_t(\mathbf{w})^2]
\le d/(1-\gamma)$. With $\gamma=\eta d^2\le1/2$, the second-moment and
exploration terms satisfy
$\eta dT/(1-\gamma)+2\gamma T\le4d^2\eta T$.
Substituting this into the general regret bound yields
\(
    \mathbb E[R(T)]
    \le CF\left(
        \frac{\log|V|}{\eta}+4d^2\eta T
    \right).
\)
The parameter choice in \eqref{eq:bsxp-parameters} yields
the square-root bound. In particular, exploration is paid for
through $2\gamma T$, while the learner updates after every
round rather than waiting for a block of basis observations.

To sharpen the dependence on the effective dimension \(d\), we further optimize the exploration distribution while retaining the same loss estimator and exponential weights update. This yields tighter control of the loss estimates and permits \(\gamma=\eta d\) instead of \(\gamma=\eta d^2\). With an appropriate learning rate, the resulting regret bound becomes \(\mathbb{E}[R(T)]\le 4CF\sqrt{dT\log|V|}\), improving the bound in Theorem 2 by a factor of \(\sqrt{d}\). The optimized exploration distribution is computed once from the known payoff features; the construction and proof are provided in Appendix D.5 (Proposition 3).

We next establish a lower bound on the dependence of the regret on
the horizon. The key idea is to construct a simple RSSG instance in
which the payoff difference between two fixed defender strategies can
be represented by a symmetric random walk. Specifically, the two
possible attack intensities induce independent $\{-1,+1\}$ increments,
so that the cumulative advantage of the hindsight-optimal fixed
strategy is governed by the fluctuation of the resulting random walk.
Since this fluctuation is of order $\sqrt{T}$, the construction yields
the following lower bound.

\begin{proposition}[Lower bound in the horizon]
\label{prop:horizon-lower-bound}
For every $C\ge 1$ and $T\ge 1$, there exists an RSSG instance with
$N=2$, $K=1$, and $F=2$ such that, for every possibly randomized
defender algorithm, there exists an oblivious sequence
$\{M(t)\}_{t=1}^T$ satisfying
\(
    \sum_{k=1}^{K}\sum_{\ell=1}^{F} m_{k\ell}(t)=C,
\)
for which
\begin{equation}
    \mathbb{E}[R(T)]
    \ge
    \frac{C}{4\sqrt{3}}\sqrt{T}.
    \label{eq:horizon-lower-bound}
\end{equation}
The expectation is over the learner's internal randomization.
Moreover, the lower bound continues to hold even if the learner
receives full-information feedback after every round.
\end{proposition}

Proposition~\ref{prop:horizon-lower-bound}
and Theorem~2 together show that the $\sqrt{T}$ dependence on the
horizon is order-optimal.


\paragraph{Experimental evaluation.}
We evaluate Algorithm~\ref{alg:fpl-VaryingN} under
adversarial and stochastic follower sequences, and
Algorithm~\ref{alg:bsxp} under aggregate feedback with
time-varying attacker sequence.
The experimental settings and results are reported in
Appendix~\ref{app:simulations}.

%% file: sections/experiments_conclusion.tex
\section{Conclusion}

We study no-regret learning in repeated Stackelberg
security games with multiple followers
and time-varying attack intensities.
Under full-information feedback, an exact MILP oracle
supports FPL without explicit strategy enumeration.
Under bandit feedback, barycentric-spanner EXP2
constructs unbiased estimates of all candidate payoffs
from a single payoff per round.
In our game model, both algorithms achieve
$O(\sqrt{T})$ expected regret against adversarial
follower sequences.
These bounds contain no additional logarithmic factors
in $T$ and match the $\Omega(\sqrt{T})$ lower bound.
Numerical experiments illustrate the performance
of both algorithms.

%% file: sections/disclosures.tex

\section*{AI Use Statement}

Generative AI tools were used for language editing,
checking mathematical proofs for potential errors and
gaps, and improving the presentation of mathematical
expressions.
They also assisted with developing code for the runtime
of LP enumeration
baseline.
The main research ideas were developed by the authors,
who retain full responsibility for the mathematical
arguments, experimental results, and final manuscript.


\section*{Reproducibility Statement}

The game and feedback models, together with the round
protocols, are specified in
Sections~\ref{Problem}, \ref{MultipleF},
and \ref{sec:bsxp-model}.
Algorithms~\ref{alg:fpl-VaryingN} and \ref{alg:bsxp}
describe the learning procedures, and the complete
MILP formulations appear in
Appendix~\ref{app:oracle-details}.
Proofs of the regret guarantees, including the analysis
of the optimized exploration distribution, are provided
in Appendices~\ref{Appendix_FullInfo_Noregret}
and \ref{app:partial-proofs}.
The simulation settings and parameter choices are
described in
Appendices~\ref{Appendix_GeneralDis}--\ref{Appendix_Simulation3}.
The runtime comparison, including the computational
environment and evaluation protocol, is documented in
Appendix~\ref{app:oracle-runtime}.

%% file: sections/simulation.tex
\section{Simulations}{\label{app:simulations}}
The following experiments assess the Algorithm 1 and Algorithm 2. The guarantee for the per-round EXP2 algorithm is established theoretically in Theorem~\ref{thm:bsxp-main}.

Figure~\ref{fig:retained-main}(a) uses $C=2$, $K=2$, and $F=2$ with the adversarial attacker sequence described in Appendix~\ref{Appendix_Simulation22}. The oracle-based Follow-the-Leader baseline can be exploited by the sequence. The second baseline, \citet{chen2026lightweight}, does not account for varying attack intensity. Figure~\ref{fig:retained-main}(b) considers $N=5$, $K=5$, $F=3$, and $C=2$ with the stochastic attacker distributions described in Appendix~\ref{Appendix_GeneralDis}. The additional full-information simulations and their settings remain in Appendices~\ref{Appendix11}--\ref{Appendix_Simulation3}.

Figure~\ref{fig:retained-main}(c) considers $N=K=F=C=2$
with the alternating attacker distributions described in
Appendix~\ref{Appendix_Simulation4}.
The type and intensity distributions change every 250 rounds
while retaining the same unique expected-optimal strategy.
Barycentric-spanner EXP2 uses only the aggregate payoff
of the selected strategy.
The green markers show mean terminal regret from separate
runs at each total horizon, with parameters selected for
that horizon and fixed throughout each run.
The dashed curve gives the corresponding theoretical
$\mathcal{O}(\sqrt{T})$ upper bound.
These results illustrate learning from aggregate feedback
under changing attacker composition.

\begin{figure}[hpbt]
\centering
\begin{subfigure}[t]{0.32\linewidth}
    \centering
    \includegraphics[width=\linewidth]{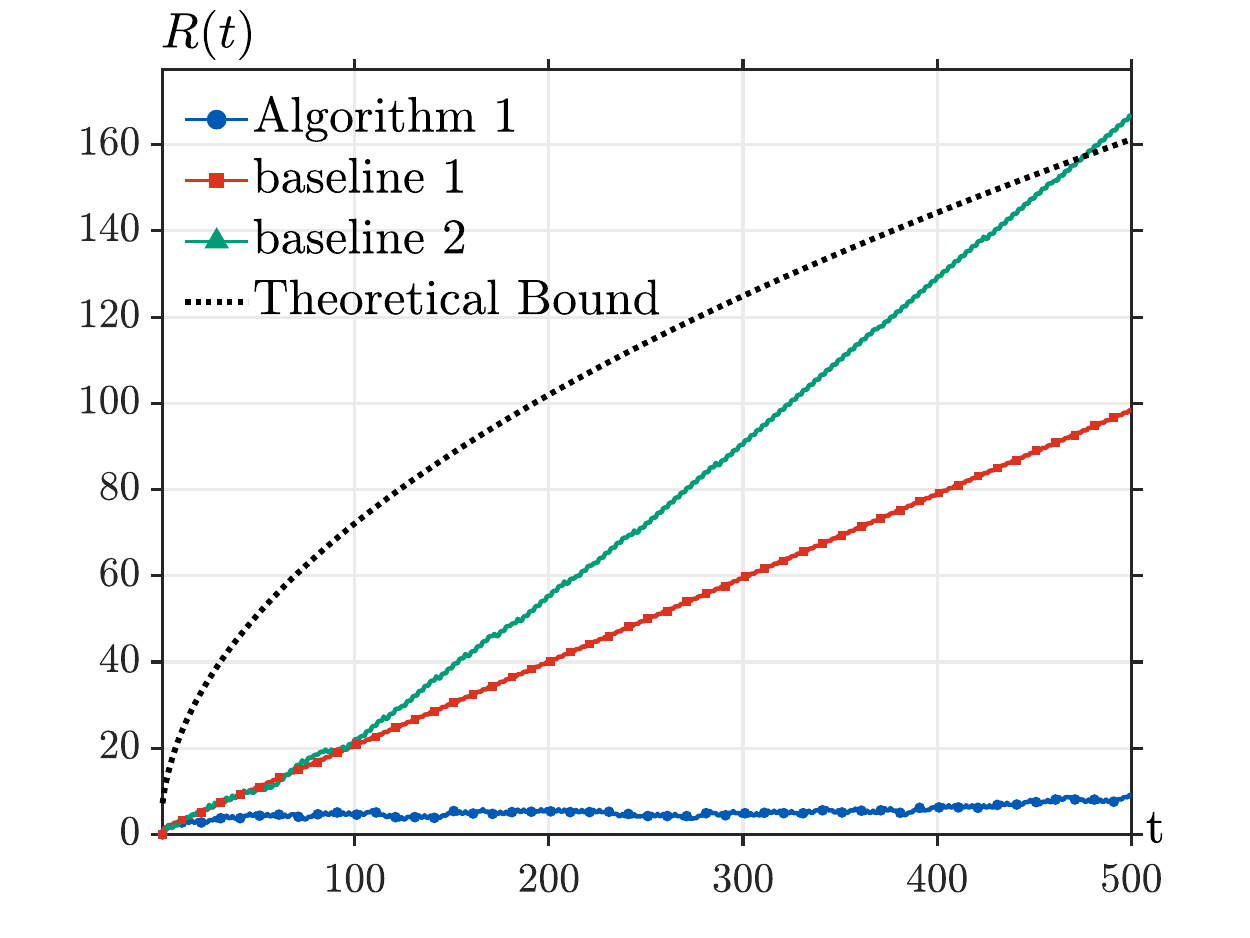}
    \caption{Multiple follower with adversarial Attacker sequence}
    \label{Simulation 1}
\end{subfigure}
\hfill
\begin{subfigure}[t]{0.32\linewidth}
    \centering
    \includegraphics[width=\linewidth]{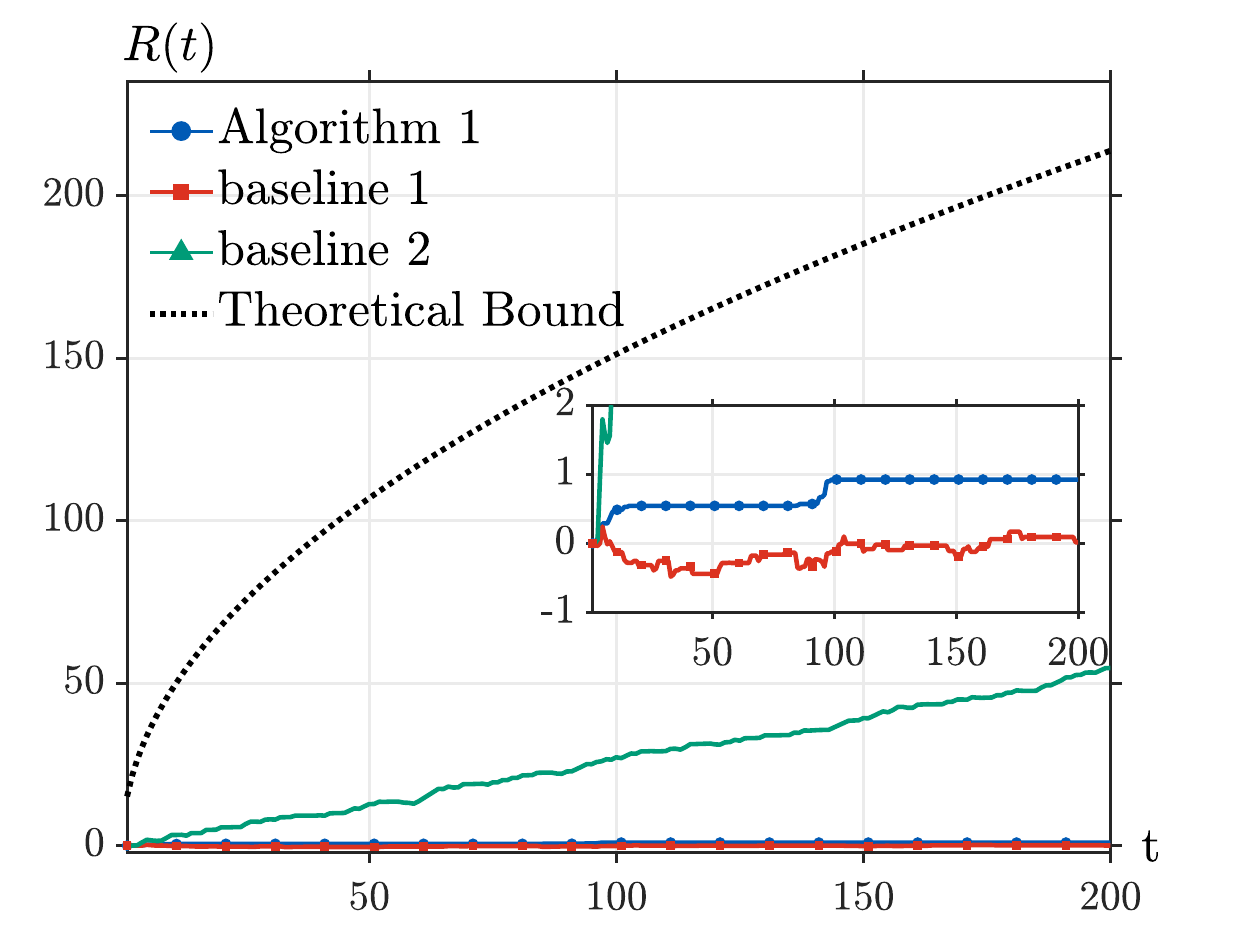}
    \caption{Multiple followers with general distribution of attackers}
    \label{Simulation 2}
\end{subfigure}
\hfill
\begin{subfigure}[t]{0.32\linewidth}
    \centering
    \includegraphics[width=\linewidth]{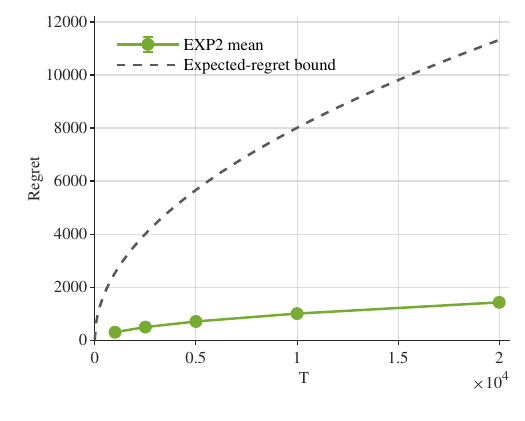}
    \caption{barycentric-spanner EXP2 under partial information feedback}
    \label{Simulation 3}
\end{subfigure}
\caption{Retained cumulative-regret plots: (a) adversarial and (b) stochastic full-information settings; (c) barycentric-spanner EXP2.}
\label{Simulation}\label{fig:retained-main}
\end{figure}

%% file: sections/oracle_complexity_proof.tex
\section{Computational complexity of the MILP oracle}\label{Appendix_MILP_Complexity}
\begin{proof}[Proof of Proposition~\ref{Milp_algorithm_complexity}]
For a fixed attacker type, the binary vectors $\mathbf h_1,\ldots,\mathbf h_F$
are nested and have cardinalities $1,\ldots,F$. Each transition adds one
previously unselected target. Hence every binary assignment is specified by
an ordered list of $F$ distinct targets, giving $N!/(N-F)!$ possibilities.
Across the $K$ types there are at most $\bigl(N!/(N-F)!\bigr)^K$ combinations.
For each combination, fix the binary variables. The remaining optimization
problem is a linear program in $\mathbf w$, the $\mathbf z_l^k$, and the
thresholds $\tau_l^k$, with $O(KFN)$ variables and constraints.
Solve each feasible linear program and choose the solution with the largest
objective value. Every feasible binary assignment has been considered, so
the resulting solution is globally optimal. This proves the enumeration
bound. A general linear program is not assumed to admit a linear-time
solver; its actual solution cost must be included in the runtime bound.
\end{proof}

%% file: sections/full_information_proof.tex
\section{Proof of the full-information regret bound}\label{Appendix_FullInfo_Noregret}

We consider time-varying follower numbers with follower number bounded by $C$ and arbitrarily attacker types from $[K]$ and $F$ attack intensities of each follower. Let the state at round $t$ be $\mathbf{M}(t)$, where each element indicates the number $m_{kl}(t)$ of active attacker $\alpha_k$ with specific attack intensity $l$.

The objective over leader decision $\mathbf{w}$ is summarized as
\[
\sum_{k=1}^{K}\sum_{l=1}^{F} \bigl(\hat{\mathbf{M}}(t) \odot \mathbf{F}(\mathbf{w})\bigr)_{kl},
\]

For each active attacker with type $\alpha_k$ attack intensity $l$, the utility term is written as $(\mathbf{F}(\mathbf{w})\bigr)_{kl}={f}_{l}^{k}(\mathbf{w})$ defined in section 4. Function $f_l^k(\mathbf{w})$ is non-convex and discontinuous. For $\mathbf{b}_k(\mathbf{w})$ with $|\mathbf{b}_k(\mathbf{w})|=l$, discontinuities appear when the attacker's best response changes. Hence $f_l^k(\mathbf{w})$ is divided into multiple regions and is piecewise linear (conditioned on $\mathbf{b}_k(\mathbf{w})$).

The defender strategy at time $t$ is modeled by
\[
  \mathbf{w}(t) = \arg\max_{\mathbf{w}}\, \sum_{k=1}^{K}\sum_{l=1}^{F} \bigl((\frac{t-1}{t}\hat{\mathbf{M}}(t-1)+\boldsymbol{\epsilon}(t)) \odot \mathbf{F}(\mathbf{w})\bigr)_{kl},,
\]
where each attacker type with specific attack intensity corresponds to a distinct piecewise-linear utility function, and $\hat{\mathbf{M}}(t-1)$ is the attack-intensity frequency (state-frequency) vector.
We have $(\mathbf{F}(\mathbf{w}))_{kl} = f_l^k(\mathbf{w}),$
and write instantaneous regret as
\[
  r(t) = \sum_{k=1}^{K}\sum_{l=1}^{F} [\bigl({\mathbf{M}}(t) \odot \mathbf{F}(\mathbf{w}^*(T))\bigr)_{kl}-\bigl({\mathbf{M}}(t) \odot \mathbf{F}(\mathbf{w}(t))\bigr)_{kl}].
\]
Therefore,
\[
  \mathbb{E}[R_T]
  = \mathbb{E}\!\left[\sum_{t=1}^{T}
  \sum_{k=1}^{K}\sum_{l=1}^{F} [\bigl({\mathbf{M}}(t) \odot \mathbf{F}(\mathbf{w}^*(T))\bigr)_{kl}-\bigl({\mathbf{M}}(t) \odot \mathbf{F}(\mathbf{w}(t))\bigr)_{kl}]\right],
\]
where equation~\eqref{fpl_TimeVarying n} is the oracle in FTPL.

\paragraph{Linear Transformation}
The online optimization problem at each round \(t\) is written as:
\[\max_{\mathbf{F(w)}}\sum_{\tau=1}^{t}
  \sum_{k=1}^{K}\sum_{l=1}^{F} \bigl({\mathbf{M}}(\tau) \odot \mathbf{F}(\mathbf{w})\bigr)_{kl}.\]
 We reformulate this problem by introducing a decision variable matrix \(\mathbf{Y} = [y_{l}^k]_{k=1,\dots,K}^{l=1,\dots,F}\), where each component $y_l^k={f}_{l}^k(\mathbf{w}) = \sum_{{\cal E}_{j} \in \mathbf{b}_k(\mathbf{w})} { \lambda_{j}^r (1-w_{j}) + \rho_{j}^r w_{j}}$ represents the expected payoff of the defender against the attacker $\alpha_k$. Stated as:

\begin{equation}\label{DOBSS-reformulated_y}
\begin{aligned}
\max_{\mathbf{Y}} \quad & \langle \mathbf{Y}, \hat{\mathbf{M}}(t) \rangle_F, \\
\text{s.t.} \quad & \mathbf{Y} \in \mathcal{Y},
\end{aligned}
\end{equation}
where $ \langle \mathbf{Y}, \hat{\mathbf{M}}(t) \rangle_F$ represents Frobenius inner product and \(\mathcal{Y} = \mathbf{F}(\mathcal{W})\) is the image of the original feasible region under the mapping \(\mathbf{F}\) defined in section 4. Although the original utility vector function of the defender \(\mathbf{F}(\mathbf{w})\) is non-linear due to the best-response structure, our reformulated objective function becomes linear in terms of the new variable \(\mathbf{Y}\) and is subject to a non-convex feasible set.

Additionally, the above linear reformulation allows us to leverage the FPL framework to ensure a sublinear regret bound. Notably, FPL achieves sublinear regret for linear objectives over any decision set including non-convex ones, contingent on the existence of efficient optimization oracles. The MILP formulation~\eqref{eq:consensus_milp_noself} just serves as such an oracle required by FPL~\citep{kalai2005efficient}. We now recall a key result from \citep{kalai2005efficient} and it is instrumental to derive the regret bound of our approach.

Let $\mathbb{E}[\mathcal{R}_T(\delta)]$ denote the expected regret of $\mathbf{Hanna}(\delta)$ up to time $T$.
Follow Hannan lead and use gradually increasing perturbations:
\begin{lemma}\citep{kalai2005efficient}
 $\mathbf{Hanna}(\delta)$: on each period $t$,
\begin{enumerate}[leftmargin=*]
    \item Choose $\boldsymbol{\epsilon}(t)$ uniformly at random from the cube
    $
{[0, {1}/{\delta\sqrt{t}}]^{KF}}$;
    \item Use strategy $\arg\max\limits_{\mathbf{Y}}\langle\mathbf{Y},\tilde{\mathbf{M}}({t})\rangle$.
\end{enumerate}
For any state matrix sequence $\mathbf{M}(1),\mathbf{M}(2),\dots$, after any number of periods $T$ $>$ 0,
\begin{align}\label{equ:expbound2}
    \mathbb{E}[\mathcal{R}_T(\delta)] \leq 2\delta R G \sqrt{T} + \frac{D \sqrt{T}}{\delta},
\end{align}
where the constants are defined as:
\begin{align*}
D &:= \sup_{\mathbf{w},\mathbf{w}'\in\mathcal{W}} \|\mathbf{F}(\mathbf{w}) - \mathbf{F}(\mathbf{w}')\|_1,\\
R & :=
    \sup_{\substack{
        \mathbf{w},\mathbf{w'}\in \mathcal{W},\;\mathbf M\ge0\\
        \sum_{k,l}m_{kl}\le C}}
    \left|\left\langle
        \mathbf M,\mathbf F(w)-\mathbf F(w')
    \right\rangle_{\mathrm F}\right|
    ,
\\
G &:= \max_{t} \|\mathbf{M}(t)\|_1.
\end{align*}
\end{lemma}

In our setting, the constants can be bounded as follows.
Since each component $f_l^k(\mathbf{w})$ lies in $[-l,l]$, we have
$|f_l^k(\mathbf{w}) - f_l^k(\mathbf{w}')| \le 2l$. Summing over all $k \in [K]$ and $l \in [F]$, it follows that
\[
D = \sup_{\mathbf{w},\mathbf{w}'\in\mathcal{W}} \|\mathbf{F}(\mathbf{w}) - \mathbf{F}(\mathbf{w}')\|_1
\le \sum_{k=1}^K \sum_{l=1}^F 2l = KF(F+1).
\]
Similarly, $|f_l^k(\mathbf{w})-f_l^k(\mathbf{w}')|\le 2F$ and each round contains
at most $C$ followers, the one-round payoff range satisfies
\[
    R :=
    \sup_{\substack{
        \mathbf{w},\mathbf{w}'\in\mathcal{W},\;\mathbf M\ge0\\
        \sum_{k,l}m_{kl}\le C}}
    \left|\left\langle
        \mathbf M,\mathbf F(w)-\mathbf F(w')
    \right\rangle_{\mathrm F}\right|
    \le 2FC.
\] 


For $G$, note that $\mathbf{M}(t)$ is a count matrix where each entry $m_{kl}(t)$ represents the number of followers choosing attacker type $k$ with attack intensity $l$ at time $t$. Since each follower selects exactly one $(k,l)$ pair, the $\ell_1$-norm of $\mathbf{M}(t)$ equals the number of followers at time $t$, denoted by $n(t)$. As $n(t) \le C$, we obtain
\[
G = \max_t \|\mathbf{M}(t)\|_1 =\max_{ t}n(t) \le C.
\]

Substituting these bounds into \eqref{equ:expbound2}, we obtain
\begin{align}\label{equ:expbound}
    \mathbb{E}[\mathcal{R}_T(\delta)] \le 4\delta FC^2 \sqrt{T} + \frac{KF(F+1)\sqrt{T}}{\delta}.
\end{align}
Choosing
\[
    \delta=\frac{\sqrt{K(F+1)}}{2C}
\]
minimizes the right-hand side of \eqref{equ:expbound}.
Substituting this choice gives
\[
    \mathbb{E}[\mathcal{R}_T(\delta)]
    \le
    2CF\sqrt{K(F+1)T}
    +2CF\sqrt{K(F+1)T}
    =
    4CF\sqrt{K(F+1)T}.
\]
Thus, for fixed game parameters, the expected regret is
$\mathcal{O}(\sqrt{T})$.
Therefore, the regret is $\mathcal{O}(\sqrt{T})$, with constants depending on $D$, $R$ and $G$.

Finally, the original bi-level optimization problem can be reduced to a standard single-level MILP. This reformulation enables efficient computation of the optimal defense strategy using off-the-shelf solvers such as Gurobi or CPLEX, making real-time implementation feasible.

%% file: sections/bandit_proofs.tex
\section{Proofs for the partial-information algorithm}
\label{app:partial-proofs}

\subsection{Exact finite-strategy reduction}\label{app:bsxp-finite-proof}
\begin{proof}[Proof of Lemma~\ref{lem:bsxp-finite}]
Fix $\mathbf w\in W$ and its realized response profile $\sigma$. With $\sigma$
held fixed, the cumulative defender payoff is linear on $P(\sigma)$
and is therefore maximized at a vertex $\mathbf v\in V$. At $\mathbf v$, $\sigma$
remains a best-response profile, and optimistic tie-breaking can only
improve the defender's payoff. Hence
$\sum_t r_t(\mathbf v)\ge\sum_t r_t(\mathbf w)$, proving
\eqref{eq:bsxp-finite-equality}.

For $|V|
    \le
    \binom{N+K\binom{N}{2}}{N-1}$, the region boundaries come from
the $N$ equations $w_j=0$ and at most $K\binom N2$ pairwise
attacker-utility equalities. In the $(N-1)$-dimensional affine space
$\sum_jw_j=1$, each vertex is specified by $N-1$ independent active
equalities. Counting their possible choices gives the stated bound.
\end{proof}
\subsection{Unbiasedness and second moments}
\label{app:bsxp-estimation}

Let $\mathcal F_{t-1}$ denote the learner's history before its
round-$t$ draw, and write
$\mathbb E_t[\cdot]=\mathbb E[\cdot\mid\mathcal F_{t-1}]$.
For the given oblivious sequence, $\mathbf g_t$ and $\ell_t(\mathbf w)$
are fixed, while $x_t,p_t,S_t$ are determined by the history.

\begin{lemma}[Per-round payoff estimation]
\label{lem:bsxp-estimation}
For every round of Algorithm~\ref{alg:bsxp},
\begin{align}
    S_t&\succeq(\gamma/d)I_d,
    \label{eq:bsxp-pd}\\
    \mathbb E_t[\widehat{\mathbf g}_t]&=\mathbf g_t,
    \qquad
    \mathbb E_t[\widehat\ell_t(\mathbf w)]=\ell_t(\mathbf w)
    \quad(\mathbf w\in V),
    \label{eq:bsxp-unbiased}\\
    |\widehat\ell_t(\mathbf w)|&\le d^2/\gamma
    \quad(\mathbf w\in V),
    \label{eq:bsxp-range}\\
    \mathbb E_t\!\left[
       \sum_{\mathbf w\in V}p_t(\mathbf w)\widehat\ell_t(\mathbf w)^2
    \right]&\le d,
    \label{eq:bsxp-p-moment}\\
    \mathbb E_t\!\left[
       \sum_{\mathbf w\in V}x_t(\mathbf w)\widehat\ell_t(\mathbf w)^2
    \right]&\le \frac{d}{1-\gamma}.
    \label{eq:bsxp-x-moment}
\end{align}
\end{lemma}

\begin{proof}
The identity in \eqref{eq:bsxp-second-moment} proves
\eqref{eq:bsxp-pd} because its first summand is positive semidefinite.
Using the exact loss representation and summing over the actual
sampling distribution gives
\begin{align}
    \mathbb E_t[\widehat{\mathbf g}_t]
    &=
    S_t^{-1}
       \sum_{\mathbf w\in V}p_t(\mathbf w)\mathbf c(\mathbf w)\ell_t(\mathbf w)\notag\\
    &=
    S_t^{-1}
       \sum_{\mathbf w\in V}p_t(\mathbf w)
          \mathbf c(\mathbf w)\mathbf c(\mathbf w)^\top\mathbf g_t
      =S_t^{-1}S_t\mathbf g_t
      =\mathbf g_t.
    \label{eq:bsxp-unbiased-proof}
\end{align}
Multiplication by the fixed vector $\mathbf c(\mathbf w)^\top$ proves the
second equality in \eqref{eq:bsxp-unbiased}.

By \eqref{eq:bsxp-pd},
$\|S_t^{-1}\|_{\mathrm{op}}\le d/\gamma$.
Together with $|\ell_t(\mathbf w(t))|\le1$ and
$\|\mathbf c(\mathbf w)\|_2\le\sqrt d$ from
\eqref{eq:bsxp-spanner-properties}, this gives
\begin{align}
    |\widehat\ell_t(\mathbf w)|
    &=
    |\mathbf c(\mathbf w)^\top S_t^{-1}\mathbf c(\mathbf w(t))\ell_t(\mathbf w(t))|
    \notag\\
    &\le
    \|\mathbf c(\mathbf w)\|_2\|S_t^{-1}\|_{\mathrm{op}}
    \|\mathbf c(\mathbf w(t))\|_2
    \le \sqrt d\,\frac d\gamma\,\sqrt d
    =\frac{d^2}{\gamma}.
    \label{eq:bsxp-range-proof}
\end{align}

For the second-moment estimate, first condition on the chosen
strategy and substitute the estimator:
\begin{align}
    \sum_{\mathbf w\in V}p_t(\mathbf w)\widehat\ell_t(\mathbf w)^2
    &=
    \widehat{\mathbf g}_t^\top
       \left(\sum_{\mathbf w\in V}p_t(\mathbf w)\mathbf c(\mathbf w)\mathbf c(\mathbf w)^\top\right)
       \widehat{\mathbf g}_t\notag\\
    &=
    \widehat{\mathbf g}_t^\top S_t\widehat{\mathbf g}_t\notag\\
    &=
    \ell_t(\mathbf w(t))^2
       \mathbf c(\mathbf w(t))^\top S_t^{-1}\mathbf c(\mathbf w(t)).
    \label{eq:bsxp-second-moment-proof-1}
\end{align}
The scalar square is at most one. The identity
$z^\top A z=\operatorname{tr}(Azz^\top)$ and the definition of $S_t$
then imply
\begin{align}
    \mathbb E_t\!\left[
       \sum_{\mathbf w\in V}p_t(\mathbf w)\widehat\ell_t(\mathbf w)^2
    \right]
    &\le
    \mathbb E_t[
       \mathbf c(\mathbf w(t))^\top S_t^{-1}\mathbf c(\mathbf w(t))]\notag\\
    &=
    \operatorname{tr}\!\left(
       S_t^{-1}\mathbb E_t[
          \mathbf c(\mathbf w(t))\mathbf c(\mathbf w(t))^\top]\right)\notag\\
    &=
    \operatorname{tr}(S_t^{-1}S_t)=d.
    \label{eq:bsxp-second-moment-proof-2}
\end{align}
Finally, $p_t(\mathbf w)\ge(1-\gamma)x_t(\mathbf w)$ and
$\widehat\ell_t(\mathbf w)^2\ge0$. Multiplying this pointwise probability
inequality by the squared estimates, summing, and taking conditional
expectations yields \eqref{eq:bsxp-x-moment}.
\end{proof}

\subsection{The exponential-weights inequality}
\label{app:bsxp-ew}

\begin{lemma}[Exponential weights with signed loss estimates]
\label{lem:bsxp-ew}
Let $x_1$ be uniform on $V$, and let
\begin{equation}
    x_{t+1}(\mathbf w)=
    \frac{x_t(\mathbf w)e^{-\eta\widehat\ell_t(\mathbf w)}}
         {\sum_{\mathbf v\in V}x_t(\mathbf v)e^{-\eta\widehat\ell_t(\mathbf v)}}.
    \label{eq:bsxp-ew-update}
\end{equation}
If $|\eta\widehat\ell_t(\mathbf w)|\le1$ for every $t,\mathbf w$, then for every
$\mathbf w^\star\in V$, pathwise,
\begin{equation}
\begin{split}
    \sum_{t=1}^{T}
       \left(\sum_{\mathbf w\in V}x_t(\mathbf w)\widehat\ell_t(\mathbf w)
             -\widehat\ell_t(\mathbf w^\star)\right)
    \le
    \frac{\log|V|}{\eta}
    +\eta\sum_{t=1}^{T}\sum_{\mathbf w\in V}
          x_t(\mathbf w)\widehat\ell_t(\mathbf w)^2.
\end{split}
\label{eq:bsxp-ew-bound}
\end{equation}
\end{lemma}
\begin{proof}
We want to compare the average estimated loss under $x_t$ with the
estimated loss of a fixed strategy $\mathbf w^\star$.
The update formula provides a starting point: the numerator is available to contain
the estimated loss of $\mathbf w^\star$, while the denominator contains the
estimated losses of all strategies, weighted by $x_t$.
Taking logarithms separates these two parts and brings
$-\eta\widehat\ell_t(\mathbf w^\star)$ out of the exponential.
To obtain the weighted average of the estimated losses from the
denominator, we use $e^{-z}\le1-z+z^2$ followed by
$\log u\le u-1$.

Since $|\eta\widehat\ell_t(\mathbf w)|\le1$, applying
$e^{-z}\le1-z+z^2$ gives
\begin{align}
    \sum_{\mathbf w\in V}x_t(\mathbf w)e^{-\eta\widehat\ell_t(\mathbf w)}
    &\le
    \sum_{\mathbf w\in V}x_t(\mathbf w)
    \left(1-\eta\widehat\ell_t(\mathbf w)
    +\eta^2\widehat\ell_t(\mathbf w)^2\right)
    \notag\\
    &=1-\eta\sum_{\mathbf w\in V}x_t(\mathbf w)\widehat\ell_t(\mathbf w)
    +\eta^2\sum_{\mathbf w\in V}x_t(\mathbf w)\widehat\ell_t(\mathbf w)^2.
\end{align}
Then, using $\log u\le u-1$ for $u>0$, we obtain
\begin{align}
    \log\left(\sum_{\mathbf w\in V}
    x_t(\mathbf w)e^{-\eta\widehat\ell_t(\mathbf w)}\right)
    &\le
    \sum_{\mathbf w\in V}x_t(\mathbf w)e^{-\eta\widehat\ell_t(\mathbf w)}-1
    \notag\\
    &\le
    -\eta\sum_{\mathbf w\in V}x_t(\mathbf w)\widehat\ell_t(\mathbf w)
    +\eta^2\sum_{\mathbf w\in V}x_t(\mathbf w)\widehat\ell_t(\mathbf w)^2.
    \label{eq:bsxp-log-normalizer}
\end{align}

The update \eqref{eq:bsxp-ew-update} also gives
\begin{equation}
    \log\frac{x_{t+1}(\mathbf w^\star)}{x_t(\mathbf w^\star)}
    =
    -\eta\widehat\ell_t(\mathbf w^\star)
    -\log\left(\sum_{\mathbf w\in V}
       x_t(\mathbf w)e^{-\eta\widehat\ell_t(\mathbf w)}\right).
    \label{eq:bsxp-weight-ratio}
\end{equation}
Combining the two displays and rearranging gives
\begin{equation}
\begin{split}
    \sum_{\mathbf w\in V}x_t(\mathbf w)\widehat\ell_t(\mathbf w)
        -\widehat\ell_t(\mathbf w^\star)
    \le
    \frac1\eta\log\frac{x_{t+1}(\mathbf w^\star)}{x_t(\mathbf w^\star)}
    +\eta\sum_{\mathbf w\in V}x_t(\mathbf w)\widehat\ell_t(\mathbf w)^2.
\end{split}
\label{eq:bsxp-one-round-ew}
\end{equation}
We now sum this inequality over all rounds, keeping the same
comparison strategy $\mathbf w^\star$. Summing over $t$ telescopes the logarithms. Since
$x_1(\mathbf w^\star)=1/|V|$ and $x_{T+1}(\mathbf w^\star)\le1$,
\[
    \sum_{t=1}^{T}
       \log\frac{x_{t+1}(\mathbf w^\star)}{x_t(\mathbf w^\star)}
       =
       \log\frac{x_{T+1}(\mathbf w^\star)}{x_1(\mathbf w^\star)}
       \le\log|V|.
\]
Substitution proves \eqref{eq:bsxp-ew-bound}.

Although the loss estimates depend on the strategies sampled by
the algorithm, the calculations above hold for any realized
sequence of estimates satisfying the stated bound.
Thus, the inequality holds for every possible sampling path,
without taking expectations or requiring independence across rounds.
\end{proof}

\subsection{Proof of the regret theorem}
\label{app:bsxp-main}

\begin{proof}[Proof of Theorem~\ref{thm:bsxp-main}]
First consider $d\ge1$ and $|V|\ge2$. Because
$\eta d^2\le\gamma$, Lemma~\ref{lem:bsxp-estimation} ensures
$|\eta\widehat\ell_t(\mathbf w)|\le1$. Thus
Lemma~\ref{lem:bsxp-ew} applies to every realization of the algorithm.
For the given oblivious sequence, choose
\[
    \mathbf w^\star\in\arg\min_{\mathbf w\in V}\sum_{t=1}^{T}\ell_t(\mathbf w).
\]
This comparison strategy is fixed independently of the learner's
random draws. By Lemma~\ref{lem:bsxp-finite}, it is also optimal
for the original continuous benchmark.

\paragraph{Step 1: Replace estimated losses by true losses.}
Conditional on $\mathcal F_{t-1}$, the distribution $x_t$ is fixed.
Lemma~\ref{lem:bsxp-estimation} therefore implies
\begin{equation}
\begin{split}
    \mathbb E_t\!\left[
       \sum_{\mathbf w\in V}x_t(\mathbf w)\widehat\ell_t(\mathbf w)
       -\widehat\ell_t(\mathbf w^\star)\right]
    =
    \sum_{\mathbf w\in V}x_t(\mathbf w)\ell_t(\mathbf w)-\ell_t(\mathbf w^\star).
\end{split}
\label{eq:bsxp-true-loss-conversion}
\end{equation}
Taking full expectations and summing over $t$, the tower property gives
\begin{align*}
    &\mathbb E\!\left[
       \sum_{t=1}^{T}
       \left(\sum_{\mathbf w\in V}x_t(\mathbf w)\widehat\ell_t(\mathbf w)
       -\widehat\ell_t(\mathbf w^\star)\right)\right]\\
    &\quad=
    \sum_{t=1}^{T}\mathbb E\!\left[
       \mathbb E_t\!\left[
       \sum_{\mathbf w\in V}x_t(\mathbf w)\widehat\ell_t(\mathbf w)
       -\widehat\ell_t(\mathbf w^\star)\right]\right]\\
    &\quad=
    \mathbb E\!\left[
       \sum_{t=1}^{T}
       \left(\sum_{\mathbf w\in V}x_t(\mathbf w)\ell_t(\mathbf w)
       -\ell_t(\mathbf w^\star)\right)\right].
\end{align*}
The first equality follows from linearity of expectation and the tower
property: for each $t$, we average first over the current draw conditional
on $\mathcal F_{t-1}$ and then over all histories up to round $t-1$.
This yields the same expectation as averaging over the full trajectory
through round $T$, since the round-$t$ term is determined by the history
up to round $t$ and is unaffected by subsequent draws.
No independence across rounds is required.
The second equality follows from conditional unbiasedness and the fact
that $x_t$ is fixed conditional on $\mathcal F_{t-1}$.

Similarly, applying the tower property and
\eqref{eq:bsxp-x-moment} yields
\begin{align*}
    \mathbb E\!\left[
       \sum_{t=1}^{T}\sum_{\mathbf w\in V}
       x_t(\mathbf w)\widehat\ell_t(\mathbf w)^2\right]
    &=
    \sum_{t=1}^{T}\mathbb E\!\left[
       \mathbb E_t\!\left[
       \sum_{\mathbf w\in V}x_t(\mathbf w)\widehat\ell_t(\mathbf w)^2
       \right]\right]
    \le \frac{dT}{1-\gamma}.
\end{align*}
Taking expectations in \eqref{eq:bsxp-ew-bound} and substituting
these relations gives
\begin{equation}
    \mathbb E\!\left[
       \sum_{t=1}^{T}
       \left(\sum_{\mathbf w\in V}x_t(\mathbf w)\ell_t(\mathbf w)
       -\ell_t(\mathbf w^\star)\right)\right]
    \le
    \frac{\log|V|}{\eta}+\frac{\eta dT}{1-\gamma}.
    \label{eq:bsxp-exploitation-regret}
\end{equation}


\paragraph{Step 2: Account for actual sampling and exploration.}
Since $\mathbf w(t)\sim p_t$,
$\mathbb E_t[\ell_t(\mathbf w(t))]=\sum_{\mathbf w\in V}p_t(\mathbf w)\ell_t(\mathbf w)$.
Furthermore,
\begin{align}
    \sum_{\mathbf w\in V}(p_t(\mathbf w)-x_t(\mathbf w))\ell_t(\mathbf w)
    &=
    \gamma\left(
       \sum_{\mathbf w\in V}\mu(\mathbf w)\ell_t(\mathbf w)
       -\sum_{\mathbf w\in V}x_t(\mathbf w)\ell_t(\mathbf w)\right)\notag\\
    &\le2\gamma.
    \label{eq:bsxp-exploration-cost}
\end{align}
The last inequality uses $\ell_t(\mathbf w)\in[-1,1]$ and the fact that
$\mu,x_t$ are probability distributions. 

To use the bound already obtained for $x_t$, we add and subtract
its average loss in each round:
\begin{align*}
    \mathbb E_t[\ell_t(\mathbf w(t))]-\ell_t(\mathbf w^\star)
    &=
    \sum_{\mathbf w\in V}p_t(\mathbf w)\ell_t(\mathbf w)-\ell_t(\mathbf w^\star)\\
    &=
    \sum_{\mathbf w\in V}x_t(\mathbf w)\ell_t(\mathbf w)-\ell_t(\mathbf w^\star)
    +\sum_{\mathbf w\in V}(p_t(\mathbf w)-x_t(\mathbf w))\ell_t(\mathbf w)\\
    &\le
    \sum_{\mathbf w\in V}x_t(\mathbf w)\ell_t(\mathbf w)-\ell_t(\mathbf w^\star)
    +2\gamma,
\end{align*}
where the last inequality follows from
\eqref{eq:bsxp-exploration-cost}.
Thus, the actual sampling distribution increases the conditional
expected loss by at most $2\gamma$ per round compared with $x_t$.

Taking full expectations and summing over $t$, the tower property gives
\begin{align}
    &\mathbb E\!\left[
       \sum_{t=1}^{T}\ell_t(\mathbf w(t))
       -\sum_{t=1}^{T}\ell_t(\mathbf w^\star)\right]
    \notag\\
    &\quad=
    \sum_{t=1}^{T}\mathbb E\!\left[
       \mathbb E_t[\ell_t(\mathbf w(t))]-\ell_t(\mathbf w^\star)\right]
    \notag\\
    &\quad\le
    \mathbb E\!\left[
       \sum_{t=1}^{T}
       \left(\sum_{\mathbf w\in V}x_t(\mathbf w)\ell_t(\mathbf w)
       -\ell_t(\mathbf w^\star)\right)\right]
    +2\gamma T
    \notag\\
    &\quad\le
    \frac{\log|V|}{\eta}
    +\frac{\eta dT}{1-\gamma}
    +2\gamma T,
    \label{eq:bsxp-normalized-regret}
\end{align}
where the final inequality uses
\eqref{eq:bsxp-exploitation-regret}.

By normalization and exact finite reduction,
\begin{equation}
    \mathbb E[R(T)]
    =
    CF\,\mathbb E\left[
       \sum_{t=1}^{T}\ell_t(\mathbf w(t))
       -\sum_{t=1}^{T}\ell_t(\mathbf w^\star)\right].
    \label{eq:bsxp-scale-back}
\end{equation}
This proves general regret bound in theorem \ref{thm:bsxp-main}.

\paragraph{Step 3: Choose parameters and retain the original scale.}
Under \eqref{eq:bsxp-parameters},
$\gamma=\eta d^2\le1/2$. Therefore
\begin{equation}
    \frac{\mathbb E[R(T)]}{CF}
    \le
       \frac{\log|V|}{\eta}+2\eta dT+2\eta d^2T
    \le
       \frac{\log|V|}{\eta}+4\eta d^2T.
    \label{eq:bsxp-parameter-bound}
\end{equation}
If $T\ge d^2\log|V|$, then
$\eta=(2d)^{-1}\sqrt{\log|V|/T}$. Both terms on the right
of \eqref{eq:bsxp-parameter-bound} equal
$2d\sqrt{T\log|V|}$, proving
$\mathbb E[R(T)]\le4CFd\sqrt{T\log|V|}$.
If $T<d^2\log|V|$, the bound $|r_t(\mathbf w)|\le CF$ gives instead
\[
    \mathbb E[R(T)]\le2CFT
       <2CFd\sqrt{T\log|V|},
\]
which is stronger than the stated bound. Finally,
$d\le KF$ and $|V|
    \le
    \binom{N+K\binom{N}{2}}{N-1}$ give the last inequality.

Their cumulative payoff equals
the continuous optimum by Lemma~\ref{lem:bsxp-finite}.
\end{proof}

\subsection{Optimized exploration with a fixed exploration rate}
\label{app:bsxp-opt}

We retain the payoff coordinates $\mathbf c(\mathbf w)$, the loss estimator,
and the exponential-weights update of Algorithm~\ref{alg:bsxp}.
Following the exploration principle of Bubeck et al.\ (2012), we
optimize the exploration distribution in the existing payoff coordinates.
Let $\Delta(V)$ denote the probability distributions on $V$, and set
\begin{equation}
\begin{split}
    \mu&\in\arg\max_{\mu'\in\Delta(V)}
       \log\det\!\left(
          \sum_{\mathbf w\in V}\mu'(\mathbf w)\mathbf c(\mathbf w)\mathbf c(\mathbf w)^\top
       \right),\\
    S_\mu&=\sum_{\mathbf w\in V}\mu(\mathbf w)\mathbf c(\mathbf w)\mathbf c(\mathbf w)^\top.
\end{split}
\label{eq:bsxp-opt-design}
\end{equation}
The objective is assigned the value $-\infty$ at singular matrices.
This distribution is computed once from the known payoff features.

\begin{proposition}[Optimized exploration with a fixed exploration rate]
\label{prop:bsxp-opt}
Retain the model, fixed oblivious sequence, and finite candidate set
of Theorem~\ref{thm:bsxp-main}, with $d\ge1$ and $|V|\ge2$.
Fix $\eta>0$ with $\eta d\le1/2$, and set $\gamma=\eta d$.
Modify the exploration rule of Algorithm~\ref{alg:bsxp} by using
\begin{equation}
    p_t=(1-\gamma)x_t+\gamma\mu,
    \qquad
    S_t=\sum_{\mathbf w\in V}p_t(\mathbf w)\mathbf c(\mathbf w)\mathbf c(\mathbf w)^\top.
    \label{eq:bsxp-opt-sampling}
\end{equation}
Both $\eta$ and $\gamma$ remain fixed throughout the run.
With the original estimator and the unchanged update
\eqref{eq:bsxp-ew-update}, the modified algorithm satisfies
\begin{equation}
\begin{split}
    \mathbb E[R(T)]
    &\le CF\left(
       \frac{\log|V|}{\eta}
       +\frac{\eta dT}{1-\gamma}+2\gamma T
       \right)\\
    &\le CF\left(\frac{\log|V|}{\eta}+4\eta dT\right).
\end{split}
\label{eq:bsxp-opt-general-bound}
\end{equation}
In particular, choosing
\begin{equation}
    \eta=\min\left\{
       \frac1{2d},\,
       \frac12\sqrt{\frac{\log|V|}{dT}}
    \right\},
    \qquad
    \gamma=\eta d
    \label{eq:bsxp-opt-parameters}
\end{equation}
gives, for every $T\ge1$,
\begin{equation}
    \mathbb E[R(T)]\le4CF\sqrt{dT\log|V|}.
    \label{eq:bsxp-opt-bound}
\end{equation}
\end{proposition}

\begin{proof}
The proof reuses the unbiasedness, second-moment, and
exponential-weights calculations above. The key additional step
is to obtain a sharper bound on the loss estimates from the
optimized exploration distribution.

Since $\mathbf c(\mathbf w_i)=e_i$, the uniform distribution on the basis
strategies has second moment $I_d/d$. The determinant is continuous
on the compact probability simplex, and its maximum is therefore
attained and strictly positive. Thus \eqref{eq:bsxp-opt-design}
has a maximizer with $S_\mu\succ0$.

For any $\mathbf w\in V$, moving an $\varepsilon$ fraction of the
exploration distribution to $\mathbf w$, with $0\le\varepsilon<1$,
gives a feasible distribution. Since $\mu$ maximizes the
log-determinant, the corresponding right derivative is
nonpositive. By Jacobi's formula,
\begin{align}
    0
    &\ge
    \left.\frac{\mathrm d}{\mathrm d\varepsilon}
       \log\det\!\left(
          (1-\varepsilon)S_\mu
          +\varepsilon\mathbf c(\mathbf w)\mathbf c(\mathbf w)^\top
       \right)\right|_{\varepsilon=0^+}\notag\\
    &=
    \operatorname{tr}\!\left[
       S_\mu^{-1}
       \bigl(\mathbf c(\mathbf w)\mathbf c(\mathbf w)^\top-S_\mu\bigr)
    \right]\notag\\
    &=
    \mathbf c(\mathbf w)^\top S_\mu^{-1}\mathbf c(\mathbf w)-d.
    \label{eq:bsxp-opt-design-proof}
\end{align}
The last equality follows from
$\operatorname{tr}(Azz^\top)=z^\top Az$,
$S_\mu^{-1}S_\mu=I_d$, and $\operatorname{tr}(I_d)=d$.
Thus,
$\mathbf c(\mathbf w)^\top S_\mu^{-1}\mathbf c(\mathbf w)\le d$
for every $\mathbf w\in V$.

Since $0<\gamma=\eta d\le1/2$, the sampling distribution gives
\[
    S_t
    =
    (1-\gamma)\sum_{\mathbf v\in V}
       x_t(\mathbf v)\mathbf c(\mathbf v)\mathbf c(\mathbf v)^\top
    +\gamma S_\mu
    \succeq\gamma S_\mu
    \succ0.
\]
Inversion reverses the order of positive-definite matrices, so
\[
    S_t^{-1}\preceq\frac1\gamma S_\mu^{-1}.
\]
Consequently, for every $\mathbf w\in V$,
\begin{equation}
    \eta\mathbf c(\mathbf w)^\top S_t^{-1}\mathbf c(\mathbf w)
    \le\frac\eta\gamma
       \mathbf c(\mathbf w)^\top S_\mu^{-1}\mathbf c(\mathbf w)
    \le\frac{\eta d}{\gamma}=1.
    \label{eq:bsxp-opt-certificate}
\end{equation}
Using $|\ell_t(\mathbf w(t))|\le1$ and the Cauchy--Schwarz inequality,
we obtain
\begin{align}
    |\widehat\ell_t(\mathbf w)|
    &=\left|\mathbf c(\mathbf w)^\top S_t^{-1}
       \mathbf c(\mathbf w(t))\ell_t(\mathbf w(t))\right|\notag\\
    &\le
    \sqrt{\mathbf c(\mathbf w)^\top S_t^{-1}\mathbf c(\mathbf w)}
    \sqrt{\mathbf c(\mathbf w(t))^\top S_t^{-1}\mathbf c(\mathbf w(t))}
    \le\frac1\eta.
    \label{eq:bsxp-opt-range}
\end{align}
Hence Lemma~\ref{lem:bsxp-ew} applies without modification.

Conditional on $\mathcal F_{t-1}$, the distributions $x_t,p_t$
and the matrix $S_t$ are fixed. The identities in
\eqref{eq:bsxp-unbiased-proof} and
\eqref{eq:bsxp-second-moment-proof-1}--\eqref{eq:bsxp-second-moment-proof-2}
use only the actual sampling distribution and its invertible
second-moment matrix. They therefore remain valid and give
\begin{equation}
\begin{split}
    \mathbb E_t[\widehat\ell_t(\mathbf w)]&=\ell_t(\mathbf w),\\
    \mathbb E_t\!\left[
       \sum_{\mathbf w\in V}x_t(\mathbf w)\widehat\ell_t(\mathbf w)^2
    \right]&\le\frac{d}{1-\gamma}.
\end{split}
\label{eq:bsxp-opt-moments}
\end{equation}
The second inequality also uses
$p_t(\mathbf w)\ge(1-\gamma)x_t(\mathbf w)$.

Fix $\mathbf w^\star\in\arg\min_{\mathbf w\in V}\sum_{t=1}^T\ell_t(\mathbf w)$.
It is independent of the learner's randomization because the
attacker sequence is fixed. Taking expectations in
\eqref{eq:bsxp-ew-bound} and using
\eqref{eq:bsxp-opt-moments} yields
\begin{equation}
    \mathbb E\sum_{t=1}^T
       \left(\sum_{\mathbf w\in V}x_t(\mathbf w)\ell_t(\mathbf w)-\ell_t(\mathbf w^\star)\right)
    \le\frac{\log|V|}{\eta}
       +\frac{\eta dT}{1-\gamma}.
    \label{eq:bsxp-opt-exploitation}
\end{equation}

Since $p_t=(1-\gamma)x_t+\gamma\mu$ and
$\ell_t(\mathbf w)\in[-1,1]$, the conditional exploration cost satisfies
\begin{align*}
    \sum_{\mathbf w\in V}(p_t(\mathbf w)-x_t(\mathbf w))\ell_t(\mathbf w)
    &=
    \gamma\left(
       \sum_{\mathbf w\in V}\mu(\mathbf w)\ell_t(\mathbf w)
       -\sum_{\mathbf w\in V}x_t(\mathbf w)\ell_t(\mathbf w)
    \right)\\
    &\le2\gamma.
\end{align*}
Using $\mathbf w(t)\sim p_t$, the scale conversion
\eqref{eq:bsxp-scale-back}, and
\eqref{eq:bsxp-opt-exploitation}, we obtain
\begin{align*}
    \frac{\mathbb E[R(T)]}{CF}
    &=
    \mathbb E\sum_{t=1}^T
       \left(
          \sum_{\mathbf w\in V}p_t(\mathbf w)\ell_t(\mathbf w)-\ell_t(\mathbf w^\star)
       \right)\\
    &\le
    \mathbb E\sum_{t=1}^T
       \left(
          \sum_{\mathbf w\in V}x_t(\mathbf w)\ell_t(\mathbf w)-\ell_t(\mathbf w^\star)
       \right)
    +2\gamma T\\
    &\le
    \frac{\log|V|}{\eta}
    +\frac{\eta dT}{1-\gamma}
    +2\gamma T.
\end{align*}
Multiplying by $CF$ proves the first inequality in
\eqref{eq:bsxp-opt-general-bound}. Since $\gamma=\eta d\le1/2$,
\[
    \frac{\eta dT}{1-\gamma}+2\gamma T
    \le2\eta dT+2\eta dT
    =4\eta dT,
\]
which proves its second inequality.

For $T\ge d\log|V|$, the choice in
\eqref{eq:bsxp-opt-parameters} gives
\[
    \frac{\log|V|}{\eta}=4\eta dT=2\sqrt{dT\log|V|}.
\]
For $T<d\log|V|$, we instead use $|r_t(\mathbf w)|\le CF$ to obtain
\[
    \mathbb E[R(T)]\le2CFT<2CF\sqrt{dT\log|V|}.
\]
This proves \eqref{eq:bsxp-opt-bound} for every $T\ge1$.
\end{proof}

\section{Proof of the horizon lower bound}
\label{app:horizon-lower-bound}

\begin{proof}[Proof of Proposition~\ref{prop:horizon-lower-bound}]
We construct a fixed RSSG instance for which every online algorithm
incurs expected regret of order $\sqrt{T}$ on some oblivious sequence.

\paragraph{Step 1: A two-target RSSG instance.}

Consider
\[
    N=2,\qquad K=1,\qquad F=2,
\]
and let the defender strategy be
\[
    \mathbf{w}=(w_1,w_2),
    \qquad
    w_1+w_2=1,
    \qquad
    w_1,w_2\ge 0.
\]

Choose the defender utilities as
\[
    \lambda_1^r=-\frac14,
    \qquad
    \rho_1^r=\frac14,
    \qquad
    \lambda_2^r=-\frac12,
    \qquad
    \rho_2^r=1.
\]
There is only one attacker type, whose utilities are
\[
    \lambda_1^1=\lambda_2^1=1,
    \qquad
    \rho_1^1=\rho_2^1=0.
\]
All of these utilities satisfy the ranges assumed in the model.

The attacker's utilities for the two targets are
\[
    u_a^1(E_1,\mathbf{w})=1-w_1,
    \qquad
    u_a^1(E_2,\mathbf{w})=1-w_2=w_1.
\]
Hence, when the attack intensity is one, the attacker chooses
target~$1$ if $w_1<1/2$ and target~$2$ if $w_1>1/2$.

At $w_1=w_2=1/2$, the attacker is indifferent between the two
targets. The corresponding defender utilities are
\[
    U(\mathbf{w},1)
    =
    -\frac14+\frac12 w_1
    =
    0
\]
and
\[
    U(\mathbf{w},2)
    =
    -\frac12(1-w_2)+w_2
    =
    \frac14.
\]
Therefore, the optimistic tie-breaking rule selects target~$2$.

It follows that the defender payoff against an intensity-one
attacker is
\begin{equation}
    f_1^1(\mathbf{w})
    =
    \begin{cases}
        -\dfrac14+\dfrac12 w_1,
        & 0\le w_1<\dfrac12,\\[2mm]
        1-\dfrac32 w_1,
        & \dfrac12\le w_1\le 1.
    \end{cases}
    \label{eq:lower-bound-f1}
\end{equation}
When the attack intensity is two, both targets are attacked, and
therefore
\begin{align}
    f_2^1(\mathbf{w})
    &=
    \left(-\frac14+\frac12 w_1\right)
    +
    \left(1-\frac32 w_1\right)
    \nonumber\\
    &=
    \frac34-w_1.
    \label{eq:lower-bound-f2}
\end{align}

\paragraph{Step 2: Two fixed defender strategies.}

Consider
\[
    \mathbf{w}^{A}=(0,1),
    \qquad
    \mathbf{w}^{B}
    =
    \left(\frac12,\frac12\right).
\]
Equations~\eqref{eq:lower-bound-f1}--\eqref{eq:lower-bound-f2}
give
\begin{equation}
\begin{array}{c|cc}
 & \text{intensity }1 & \text{intensity }2\\
\hline
\mathbf{w}^{A}
    & -\dfrac14 & \dfrac34\\[1mm]
\mathbf{w}^{B}
    & \dfrac14 & \dfrac14.
\end{array}
\label{eq:lower-bound-payoff-table}
\end{equation}

We next verify that no other feasible defender strategy has larger
expected payoff when the two attack intensities occur with equal
probability. From
\eqref{eq:lower-bound-f1}--\eqref{eq:lower-bound-f2},
\begin{equation}
    \frac12 f_1^1(\mathbf{w})
    +
    \frac12 f_2^1(\mathbf{w})
    =
    \begin{cases}
        \dfrac14-\dfrac14 w_1,
        & 0\le w_1<\dfrac12,\\[2mm]
        \dfrac78-\dfrac54 w_1,
        & \dfrac12\le w_1\le 1.
    \end{cases}
    \label{eq:lower-bound-average-payoff}
\end{equation}
Since the first expression is at most $1/4$, and the second satisfies
\[
    \frac78-\frac54 w_1
    \le
    \frac78-\frac58
    =
    \frac14
\]
for $w_1\ge 1/2$, we obtain
\begin{equation}
    \frac12 f_1^1(\mathbf{w})
    +
    \frac12 f_2^1(\mathbf{w})
    \le
    \frac14
    \qquad
    \text{for every } \mathbf{w}\in W.
    \label{eq:lower-bound-all-actions}
\end{equation}

\paragraph{Step 3: A random oblivious sequence.}

We first consider a random sequence of follower configurations.
Before the game starts, independently for every round $t$, draw the
attack intensity $l_t$ uniformly from $\{1,2\}$.

There are exactly $C$ followers in every round, all with the same
realized intensity. Thus,
\[
    M(t)
    =
    \begin{cases}
        (C,0), & l_t=1,\\
        (0,C), & l_t=2.
    \end{cases}
\]
This satisfies
\[
    \sum_{k=1}^{K}\sum_{\ell=1}^{F}m_{k\ell}(t)=C
\]
in every round.

Fix an arbitrary, possibly randomized, learning algorithm.
Conditional on the entire history before round $t$ and on its
selected strategy $\mathbf{w}(t)$, the current intensity $l_t$ is
still uniform over $\{1,2\}$. Hence
\begin{align}
    \mathbb{E}
    \left[
        r_t(\mathbf{w}(t))
        \,\middle|\,
        \text{history and }\mathbf{w}(t)
    \right]
    &=
    \frac{C}{2}f_1^1(\mathbf{w}(t))
    +
    \frac{C}{2}f_2^1(\mathbf{w}(t))
    \nonumber\\
    &\le
    \frac{C}{4},
    \label{eq:lower-bound-learner-one-round}
\end{align}
where the inequality follows from
\eqref{eq:lower-bound-all-actions}.

Taking expectations and summing over $t$ gives
\begin{equation}
    \mathbb{E}
    \left[
        \sum_{t=1}^{T}r_t(\mathbf{w}(t))
    \right]
    \le
    \frac{CT}{4}.
    \label{eq:lower-bound-learner-total}
\end{equation}

Notice that this argument does not depend on what feedback is
revealed after each round. In particular,
\eqref{eq:lower-bound-learner-total} remains valid even under
full-information feedback.

\paragraph{Step 4: Payoff of the hindsight benchmark.}

Define
\[
    X_t
    =
    \begin{cases}
        -1, & l_t=1,\\
        +1, & l_t=2,
    \end{cases}
    \qquad
    S_T=\sum_{t=1}^{T}X_t.
\]
The variables $X_1,\ldots,X_T$ are independent random
variables.

From \eqref{eq:lower-bound-payoff-table},
\[
    r_t(\mathbf{w}^{A})
    =
    \frac{C}{4}
    +
    \frac{C}{2}X_t,
\]
while
\[
    r_t(\mathbf{w}^{B})
    =
    \frac{C}{4}
\]
for both possible attack intensities. Therefore,
\begin{equation}
    \sum_{t=1}^{T}r_t(\mathbf{w}^{A})
    =
    \frac{CT}{4}
    +
    \frac{C}{2}S_T,
    \label{eq:lower-bound-action-A}
\end{equation}
and
\begin{equation}
    \sum_{t=1}^{T}r_t(\mathbf{w}^{B})
    =
    \frac{CT}{4}.
    \label{eq:lower-bound-action-B}
\end{equation}

Since the hindsight benchmark maximizes over all feasible defender
strategies, it is at least as large as the better of these two:
\begin{align}
    \max_{\mathbf{w}\in W}
    \sum_{t=1}^{T}r_t(\mathbf{w})
    &\ge
    \max
    \left\{
        \sum_{t=1}^{T}r_t(\mathbf{w}^{A}),
        \sum_{t=1}^{T}r_t(\mathbf{w}^{B})
    \right\}
    \nonumber\\
    &=
    \frac{CT}{4}
    +
    \frac{C}{2}\max\{S_T,0\}.
    \label{eq:lower-bound-comparator}
\end{align}

Combining
\eqref{eq:lower-bound-learner-total}
and
\eqref{eq:lower-bound-comparator}
yields
\begin{align}
    \mathbb{E}[R(T)]
    &\ge
    \frac{C}{2}
    \mathbb{E}[\max\{S_T,0\}]
    \nonumber\\
    &=
    \frac{C}{4}\mathbb{E}|S_T|,
    \label{eq:lower-bound-random-walk}
\end{align}
where the last equality follows from the symmetry of $S_T$.

\paragraph{Step 5: A lower bound for the random walk.}

Since the $X_t$ are independent with
\[
    \mathbb{E}[X_t]=0,
    \qquad
    X_t^2=1,
\]
we have
\begin{equation}
    \mathbb{E}[S_T^2]=T.
    \label{eq:lower-bound-second-moment}
\end{equation}

For the fourth moment, all terms containing an odd power of some
$X_t$ have zero expectation. Hence
\begin{align}
    \mathbb{E}[S_T^4]
    &=
    \sum_{t=1}^{T}\mathbb{E}[X_t^4]
    +
    6
    \sum_{1\le s<t\le T}
    \mathbb{E}[X_s^2X_t^2]
    \nonumber\\
    &=
    T+6\binom{T}{2}
    \nonumber\\
    &=
    3T^2-2T
    \le
    3T^2.
    \label{eq:lower-bound-fourth-moment}
\end{align}

Applying H\"older's inequality,
\begin{align}
    \mathbb{E}[S_T^2]
    &=
    \mathbb{E}
    \left[
        |S_T|^{2/3}|S_T|^{4/3}
    \right]
    \nonumber\\
    &\le
    \left(\mathbb{E}|S_T|\right)^{2/3}
    \left(\mathbb{E}|S_T|^4\right)^{1/3}.
\end{align}
Rearranging and using
\eqref{eq:lower-bound-second-moment} and
\eqref{eq:lower-bound-fourth-moment},
\begin{align}
    \mathbb{E}|S_T|
    &\ge
    \frac{
        \left(\mathbb{E}[S_T^2]\right)^{3/2}
    }{
        \left(\mathbb{E}[S_T^4]\right)^{1/2}
    }
    \nonumber\\
    &\ge
    \frac{T^{3/2}}{\sqrt{3}T}
    =
    \sqrt{\frac{T}{3}}.
    \label{eq:lower-bound-rademacher}
\end{align}

Substituting \eqref{eq:lower-bound-rademacher} into
\eqref{eq:lower-bound-random-walk} gives
\begin{equation}
    \mathbb{E}[R(T)]
    \ge
    \frac{C}{4\sqrt{3}}\sqrt{T}.
    \label{eq:lower-bound-random-sequence}
\end{equation}

\paragraph{Step 6: From the random construction to an oblivious
deterministic sequence.}

The expectation in
\eqref{eq:lower-bound-random-sequence} is currently taken over both
the random intensity sequence and the learner's internal
randomization. Equivalently,
\[
    \mathbb{E}_{l_1,\ldots,l_T}
    \left[
        \mathbb{E}_{\mathrm{alg}}
        \left[
            R(T)
            \mid
            l_1,\ldots,l_T
        \right]
    \right]
    \ge
    \frac{C}{4\sqrt{3}}\sqrt{T}.
\]
Therefore, there exists at least one deterministic realization
$(l_1,\ldots,l_T)\in\{1,2\}^T$ such that
\[
    \mathbb{E}_{\mathrm{alg}}
    \left[
        R(T)
        \mid
        l_1,\ldots,l_T
    \right]
    \ge
    \frac{C}{4\sqrt{3}}\sqrt{T}.
\]
Fix this realization before the game starts. The corresponding
sequence $\{M(t)\}_{t=1}^T$ is therefore oblivious, and after fixing
it the only remaining expectation is over the learner's internal
randomization.

This proves
\[
    \mathbb{E}[R(T)]
    \ge
    \frac{C}{4\sqrt{3}}\sqrt{T},
\]
and completes the proof.
\end{proof}

%% file: sections/oracle_details.tex
\section{Exact Computation for Attacker Sequences}\label{app:oracle-details}


To solve this generalized problem, we extend our previous single-level MILP reduction in \ref{eq:consensus_milp_noself}. Specifically, we generalize components of the MILP: the decision variables, the linearization of bilinear terms, and the encoding of the attacker best-response conditions via Big-M constraints, to account for multiple attacker types and intensities. The resulting MILP computes the defender's exact optimal allocation given the empirical count matrix $\hat{\mathbf{M}}(t)$:

\begin{subequations}\label{fpl_TimeVarying n}
\begin{align}
&\max_{\mathbf{w}, \{\mathbf{z}_{1:F}^k\}_{k=1}^K, \{\mathbf{h}_{1:F}^k\}_{k=1}^K, \{\tau_{1:F}^k\}_{k=1}^K} \quad
\sum_{k=1}^{K} \sum_{l=1}^{F} \hat{m}_{kl} \cdot R_l^k(\mathbf{w}, \mathbf{z}_{l}^k) \label{obj:multi_expected_utility2} \\
\text{s.t.} \quad & \sum_{\mathcal{E}_{j}\in\mathcal{E}^-} w_{j} = 1, \quad w_{j} \geq 0,  \quad \forall k \in \{1,\dots,K\}, \forall l \in \{1,\dots,F\}, \forall \mathcal{E}_{j} \in \mathcal{E}^-: \label{const:simplex_weight3} \\
& R_l^k(\mathbf{w}, \mathbf{z}_{l}^k) = \sum_{\mathcal{E}_{j} \in \mathcal{E}^-} \left[ \lambda_j^r h_{j,l}^k - \lambda_j^r z_{j,l}^k + \rho_j^r z_{j,l}^k \right], \quad z_{j,l}^k \le w_j,\quad z_{j,l}^k \ge w_j + h_{j,l}^k - 1, \label{con:z_lb_k2} \\
& \sum_{\mathcal{E}_{j} \in \mathcal{E}^-} h_{j,l}^k = l, \quad \mathbf{h}_{l}^k \ge \mathbf{h}_{l-1}^k \quad (l \ge 2), \label{const:nested_logic_k2} \\
& \lambda_j^k (1-w_{j}) + \rho_j^k w_{j} \ge \tau_l^k - B (1 - h_{j,l}^k), \quad \lambda_j^k (1-w_{j}) + \rho_j^k w_{j} \le \tau_l^k + B h_{j,l}^k, \label{const:bigM_upper_k2} \\
& z_{j,l}^k \in [0,1],\quad z_{j,l}^k \le h_{j,l}^k,\quad h_{j,l}^k \in \{0,1\},\quad -1\leq\tau_l^k\leq1. \label{watever}
\end{align}
\end{subequations}

%% file: sections/oracle_runtime.tex
\section{Empirical Efficiency of the Oracle}
\label{app:oracle-runtime}

We compare the wall-clock runtime of a single MILP oracle call with
exhaustive LP enumeration. For each configuration, we generate five
independent random game instances. Within each repetition, both methods
receive the same game parameters and oracle objective weights.
The enumeration baseline considers every ordered list of $F$ distinct
targets for each of the $K$ attacker types, yielding
$\bigl(N!/(N-F)!\bigr)^K$ candidate LPs.
Timing covers the complete oracle calls, including model construction
and optimization. Enumeration time includes
all candidate LPs.

\begin{table}[t]
    \centering
    \caption{Runtime comparison between the MILP oracle and exhaustive
    LP enumeration. Times are reported as mean $\pm$ sample standard
    deviation, in seconds, over five independent repetitions.
    Speedup is the ratio of the enumeration mean to the MILP mean;
    values above one favor the MILP oracle.}
    \label{tab:oracle-runtime}
    \small
    \setlength{\tabcolsep}{4pt}
    \renewcommand{\arraystretch}{1.12}
    \begin{tabular}{rrr r cc r}
        \toprule
        $N$ & $K$ & $F$ & \# LPs
        & MILP (s) & Enumeration (s) & Speedup \\
        \midrule
        3 & 2 & 2 & 36
        & $0.240 \pm 0.015$ & $0.063 \pm 0.004$
        & $0.26\times$ \\

        4 & 2 & 2 & 144
        & $0.245 \pm 0.014$ & $0.230 \pm 0.013$
        & $0.94\times$ \\

        5 & 2 & 2 & 400
        & $0.271 \pm 0.026$ & $0.640 \pm 0.024$
        & $2.36\times$ \\

        4 & 2 & 3 & 576
        & $0.325 \pm 0.033$ & $0.910 \pm 0.016$
        & $2.80\times$ \\
        \midrule
        4 & 3 & 2 & 1,728
        & $0.333 \pm 0.026$ & $2.757 \pm 0.031$
        & $8.29\times$ \\

        5 & 3 & 2 & 8,000
        & $0.366 \pm 0.020$ & $12.807 \pm 0.072$
        & $35.00\times$ \\

        6 & 3 & 2 & 27,000
        & $0.392 \pm 0.014$ & $43.401 \pm 0.472$
        & $110.71\times$ \\

        7 & 3 & 2 & 74,088
        & $0.402 \pm 0.038$ & $119.350 \pm 0.914$
        & $296.80\times$ \\

        8 & 3 & 2 & 175,616
        & $0.436 \pm 0.047$ & $281.760 \pm 3.452$
        & $646.27\times$ \\

        9 & 3 & 2 & 373,248
        & $0.414 \pm 0.043$ & $596.190 \pm 7.053$
        & $1439.40\times$ \\
        \bottomrule
    \end{tabular}
\end{table}

All 50 paired runs completed successfully and passed the objective
agreement check
\[
    |v_{\mathrm{MILP}}-v_{\mathrm{LP}}|
    \leq 10^{-6}
    \max\{1,|v_{\mathrm{MILP}}|,|v_{\mathrm{LP}}|\},
\]
where $v_{\mathrm{MILP}}$ and $v_{\mathrm{LP}}$ denote the objective
values returned by the two methods.

Table~\ref{tab:oracle-runtime} shows that enumeration is faster for
the smallest instance and has a similar runtime at $(N,K,F)=(4,2,2)$.
Its cost then increases substantially with the number of candidate LPs.
In the fixed-$(K,F)=(3,2)$ series, increasing $N$ from 4 to 9 raises
enumeration time from 2.757 to 596.190 seconds, while the MILP mean
remains between 0.333 and 0.436 seconds.
At the largest tested configuration, $(N,K,F)=(9,3,2)$, the MILP oracle
achieves approximately a $1{,}439\times$ speedup.
These results provide empirical evidence that the MILP implementation
is efficient relative to explicit LP enumeration over the tested range.

The experiments were implemented in MATLAB R2024b Update 1 (version 24.2.0.2740171), with YALMIP R20230622 used to formulate the MILPs and Gurobi 9.0.0 used to solve both the MILPs and the LPs in the enumeration baseline.

%% file: sections/examples_experiments.tex
\section{Stochastic multi-type attacker distributions}\label{Appendix_GeneralDis}
\paragraph{Multi-Type Attack Distribution in Stackelberg Games}
We consider a Stackelberg setting involving a single leader and $n \geq 1$ followers. The leader has $N=|\mathcal{E}| \geq 2$ actions and selects a mixed strategy $\mathbf{w} \in \mathcal{W}$ over them. over these actions. Each follower has a finite action set ${\mathcal{E}}^-$ with cardinality $|{\mathcal{E}}^-| = N$, and the joint action of all followers is denoted by $\mathbf{v} = (v_1, \ldots, v_n)$.

In each round, every follower $i$ is associated with a private type $\alpha_i \in   [K]$ and an attack intensity $l_i \in \{1, \ldots, F\}$. We denote a follower of type $\alpha_k$ with attack intensity $l$ as $\theta_l^k$. The joint type profile of all followers is defined as $\boldsymbol{\theta} = (\theta_{l_1}^{\alpha_1}, \ldots,\theta_{l_n}^{\alpha_n}) \in \Theta^n$,
which is drawn from a bayesian setting distribution $\mathcal{D}$, i.e., $\boldsymbol{\theta} \sim \mathcal{D}$. In this context, we consider a general type distribution in which follower types and attack intensities may be arbitrarily correlated. Let $m_{kl}=\sum_{i=1}^n \mathbb{I}(\theta_{l_i}^{\alpha_i}=\theta_{l}^{k})$, which denotes the number of attackers of type $\alpha_k$ with attack intensity $l$ at a given time step. The utility function associated with type $\theta_l^k$ is defined as:

\begin{subequations}\label{f_l^k}
\begin{align}
    {f}_{l}^{k}(\mathbf{w}) &= \sum_{{\cal E}_{j} \in \mathbf{b}_k(\mathbf{w})} { \lambda_{j}^r (1-w_{j}) + \rho_{j}^r w_{j}}, \\ \text{s.t.} \quad &  \sum_{{\cal E}_{j}\in{{\cal E}^ -} } w_{j}   = 1, \quad \mathbf{b}_k(\mathbf{w}) = \arg \max_{{\cal{E}}_f^k \cap \mathcal{E}^-} \sum_{\mathcal{E}_{j} \in {{\cal{E}}_f^k \cap \mathcal{E}^-}} u_a^k({\cal E}_{j}, \mathbf{w}),   \quad |\mathbf{b}_k(\mathbf{w})|=l,
\end{align}
\end{subequations}
where ${\cal{E}}_f^k$ is the targets set selected by the attacker $\alpha_k$ and
${u}_a^k({\cal E}_{j}, \mathbf{w})
= { \lambda_{j}^k (1-w_{j}) + \rho_{j}^k w_{j}}$. Let $\mathbf{M} \in \mathbb{R}^{K \times F}$ be the matrix with entrie $(\mathbf{M})_{k,l} = m_{kl}$, where $m_{kl}$ denotes the number of attackers of type $k$ with intensity $l$.
Similarly, define $\mathbf{F}(\mathbf{w}) \in \mathbb{R}^{K \times F}$ with entries $(\mathbf{F}(\mathbf{w}))_{k,l} = f_l^k(\mathbf{w})$. The total utility can then be compactly expressed as the sum of element-wise products
$\sum_{k=1}^{K}\sum_{l=1}^{F} \bigl(\mathbf{M} \odot \mathbf{F}(\mathbf{w})\bigr)_{kl}$,
where $\odot$ denotes the Hadamard (entrywise) product.

Based on this, the corresponding single-level MILP for the multi-type setting is given as follows:
\begin{subequations}\label{eq:consensus_milp_multiK}
\begin{align}
&\max_{\mathbf{w}, \{\mathbf{z}_{1:F}^k\}_{k=1}^K, \{\mathbf{h}_{1:F}^k\}_{k=1}^K, \{\tau_{1:F}^k\}_{k=1}^K} \quad
\sum_{k=1}^{K} \sum_{l=1}^{F} m_{kl} \cdot R_l^k(\mathbf{w}, \mathbf{z}_{l}^k)  \\
\text{s.t.} \quad & \sum_{\mathcal{E}_{j}\in\mathcal{E}^-} w_{j} = 1, \quad w_{j} \geq 0, \quad \forall k \in \{1,\dots,K\}, \forall l \in \{1,\dots,F\}, \forall \mathcal{E}_{j} \in \mathcal{E}^-: \\
& R_l^k(\mathbf{w}, \mathbf{z}_{l}^k) = \sum_{\mathcal{E}_{j} \in \mathcal{E}^-} \left[ \lambda_j^r h_{j,l}^k - \lambda_j^r z_{j,l}^k + \rho_j^r z_{j,l}^k \right], \quad z_{j,l}^k \le w_j,\ z_{j,l}^k \ge w_j + h_{j,l}^k - 1, \label{con:z_lb_k} \\
& \sum_{\mathcal{E}_{j} \in \mathcal{E}^-} h_{j,l}^k = l, \quad\sum_{k=1}^{K}\sum_{l=1}^{F}m_{kl}=n, \quad \mathbf{h}_{l}^k \ge \mathbf{h}_{l-1}^k \ (l \ge 2), \label{const:nested_logic_k} \\
& \lambda_j^k (1-w_{j}) + \rho_j^k w_{j} \ge \tau_l^k - B (1 - h_{j,l}^k), \quad \lambda_j^k (1-w_{j}) + \rho_j^k w_{j} \le \tau_l^k + B h_{j,l}^k, \label{const:bigM_upper_k} \\
& z_{j,l}^k \in [0,1],\quad h_{j,l}^k \in \{0,1\},\quad z_{j,l}^k \le h_{j,l}^k,\quad -1\leq\tau_l^k\leq1.
\end{align}
\end{subequations}
Refer to Figure~\ref{Simulation 2}. In our experiments, the attacker and defender utility matrices are defined as follows. The $(i,j)$-th entry of each matrix represents the utility of attacker $j$ attacking target $i$.

\textbf{Attacker Utilities}

\[
U_{\text{attacker}}^u =
\begin{bmatrix}
0.55 & 0.79 & 0.70 & 0.85 & 0.27 \\
0.73 & 0.47 & 0.98 & 0.40 & 0.33 \\
0.80 & 0.46 & 0.65 & 0.41 & 0.59 \\
0.32 & 0.07 & 0.64 & 0.15 & 0.34 \\
0.89 & 0.31 & 0.84 & 0.05 & 0.20
\end{bmatrix},\]\[
U_{\text{attacker}}^c =
\begin{bmatrix}
-0.88 & -0.02 & -0.76 & -0.53 & -0.50 \\
-0.51 & -0.42 & -0.34 & -0.09 & -0.33 \\
-0.20 & -0.65 & -0.79 & -0.84 & -0.86 \\
-0.61 & -0.21 & -0.39 & -0.15 & -0.62 \\
-0.39 & -0.16 & -0.92 & -0.36 & -0.97
\end{bmatrix}.
\]

\textbf{Defender Utilities}

\[
U_{\text{defender}}^c =
\begin{bmatrix}
0.99 & 0.42 & 0.47 & 0.89 & 0.83
\end{bmatrix},\]\[
U_{\text{defender}}^u =
\begin{bmatrix}
-0.16 & -0.97 & -0.40 & -0.26 & -0.84
\end{bmatrix}.
\]

Here, $U_{\text{attacker}}^u$ and $U_{\text{attacker}}^c$ correspond to the attacker's utilities when the target is unprotected and protected, respectively. Similarly, $U_{\text{defender}}^u$ and $U_{\text{defender}}^c$ represent the defender's utilities under the same conditions.

\section{Adversarial sequences with one attacker type}\label{Appendix11}

In this section, we illustrate a concrete example for the single-follower case shown in fig \ref{Single attacker type with adversarial attack intensity sequence_Fig}, where there is only one attacker type active in each round. The adversary can design a sequence of attack intensities over multiple rounds to systematically deceive the defender. By carefully constructing this sequence, the attacker ensures that the defender, who chooses strategies based only on historical empirical frequencies, repeatedly selects suboptimal strategies. This leads to a linear accumulation of regret over time, demonstrating how adversarial attack intensity can exploit the defender's reliance on past observations.

The adversary designs a repeated SSG to induce linear regret of the defender in Table \ref{Utility Matrix of Defender}, \ref{Utility Matrix of Attacker}.

\begin{table}[ht!]
  \begin{center}
    \caption{Defender's utility matrix}
    \label{Utility Matrix of Defender}
    \begin{tabular}{@{}cccccc@{}}
      \toprule
      \textbf{\begin{tabular}[c]{@{}c@{}}Protected Target\end{tabular}} & \textbf{ Attack Target 1 } & \textbf{Attack Target 2} \\
      \midrule
      Target \ 1  & $\frac{1}{4}$ & -$\frac{1}{4}$ \\
      Target \ 2 & \(-\frac{1}{4}\) & 1 \\
      \bottomrule
    \end{tabular}
  \end{center}
\end{table}

\begin{table}[ht!]
  \begin{center}
    \caption{Attacker's utility matrix}
    \label{Utility Matrix of Attacker}
    \begin{tabular}{@{}cccccc@{}}
      \toprule
      \textbf{\begin{tabular}[c]{@{}c@{}}Protected Target\end{tabular}} & \textbf{Attack Target 1} & \textbf{Attack Target 2} \\
      \midrule
       Target \ 1  & 0 & 1 \\
       Target \ 2  & \(\frac{1}{2}\) & 0 \\
      \bottomrule
    \end{tabular}
  \end{center}
\end{table}

We can split the feasible set into two regions with different follower's response when attack intensity $l=1$:
\[
\text{region }\mathrm{1}:\ 1-w_1\ge 2w_1,
\qquad
\text{region }\mathrm{2}:\ 1-w_1\le 2w_1,
\]
with boundary point $\left(\frac{1}{3},\frac{2}{3}\right)$. In region 1, the attacker will best respond to attack target for higher expected utility. In region 2, the best response is attacking target 2.

For attack intensity $l=1$, we have
\[
\begin{aligned}
\mathbf{w}\in \text{region 1}:\quad &R_1(\mathbf{w})=-\frac{1}{4}(1-w_1)+\frac{1}{4}w_1=-\frac{1}{4}+\frac{1}{2}w_1,\\
\mathbf{w}\in \text{region 2}: \quad &R_1(\mathbf{w})=-\frac{1}{4}(1-w_2)+w_2=-\frac{1}{4}+\frac{5}{4}w_2.
\end{aligned}
\]

For attack intensity $l=2$, the note writes the combined term:
\[
R_2(\mathbf{w})=\frac{1}{4}\bigl(-(1-w_1)+w_1\bigr)+\frac{1}{4}\bigl(-(1-w_2)+4w_2\bigr).
\]

Let $q_1$ and $q_2$ be the frequencies of intensities $l=1$ and $l=2$, respectively. The expected utility of the defender is
\[
\begin{aligned}
\mathbf{w}\in \text{region 1}:\quad
\langle \mathbf{f}(\mathbf{w}), \mathbf{q} \rangle&=R_1(\mathbf{w})q_1+R_2(\mathbf{w})q_2 \\
&=\frac{1}{4}\bigl(-(1-w_1)+w_1\bigr)q_1\\&
+\frac{1}{4}\bigl(-(1-w_1)+w_1-(1-w_2)+4w_2\bigr)q_2,\\[2mm]
\mathbf{w}\in \text{region 2}:\quad
\langle \mathbf{f}(\mathbf{w}), \mathbf{q} \rangle&=R_1(\mathbf{w})q_1+R_2(\mathbf{w})q_2 \\
&=\frac{1}{4}\bigl(-(1-w_2)+4w_2\bigr)q_1\\&
+\frac{1}{4}\bigl(-(1-w_1)+w_1-(1-w_2)+4w_2\bigr)q_2.
\end{aligned}
\]

Taking the polytope/extreme-point argument into each region, we get
\[
\begin{aligned}
\text{region }\mathrm{1}:\quad
&\mathbf{w}=(0,1), &&\langle\mathbf{f}(\mathbf{w}), \mathbf{q} \rangle=-\frac{1}{4}q_1+\frac{3}{4}q_2,\\
&\mathbf{w}=\left(\frac{1}{3},\frac{2}{3}\right), &&\langle\mathbf{f}(\mathbf{w}), \mathbf{q} \rangle=-\frac{1}{12}q_1+\frac{1}{2}q_2;\\[2mm]
\text{region }\mathrm{2}:\quad
&\mathbf{w}=\left(\frac{1}{3},\frac{2}{3}\right), &&\langle\mathbf{f}(\mathbf{w}), \mathbf{q} \rangle=\frac{7}{12}q_1+\frac{1}{2}q_2,\\
&\mathbf{w}=(1,0), &&\langle\mathbf{f}(\mathbf{w}), \mathbf{q} \rangle=-\frac{1}{4}q_1.
\end{aligned}
\]

By evaluating the vertices of each region (extreme points) with tie-breaking in favor of defender, we obtain the set of candidate strategies for the defender that may be optimal in each region:
\[
\text{A}:\ \mathbf{w}=(0,1)\ \text{(region 1)},
\qquad
\text{B}:\ \mathbf{w}=\left(\frac{1}{3},\frac{2}{3}\right)\ \text{(region 2)}.
\]

By comparing A and B, we have
\[
-q_1+3q_2>\frac{7}{3}q_1+2q_2
\iff q_2>\frac{10}{3}q_1,
\]
which means A is better than B when $q_2>\frac{10}{3}q_1$; otherwise, $q_2<\frac{10}{3}q_1$,
B is better than A.

The above formulation assumes that the attack intensity is known in advance. In practice, however, its distribution is unknown and only empirical frequencies from historical data are available.  We therefore base our analysis on these empirical estimates.

\textbf{Adversarial sequence for $q_1,q_2$}

Assume that the defender in round $t$ can only observe attack intensities up to round $t-1$. Accordingly, it selects a strategy that is optimal with respect to the empirical estimates $(\hat{q}_1(t-1), \hat{q}_2(t-1))$. At each round $t$, the attacker selects an intensity $l \in \{1,2\}$.  The defender observes only the history up to round $t-1$, forms empirical frequency estimates accordingly
\[
\hat q_1({t-1})=\frac{1}{t-1}\sum_{s=1}^{t-1}\mathbb{I}\{l_s=1\},
\qquad
\hat q_2({t-1})=\frac{1}{t-1}\sum_{s=1}^{t-1}\mathbb{I}\{l_s=2\},
\]
and plays the strategy optimal for
$( \hat q_1({t-1}),\hat q_2({t-1}))$.

Recall that
\[
\langle\mathbf{f}(\mathbf{w}_A), \mathbf{q} \rangle=-\frac{1}{4}q_1+\frac{3}{4}q_2,
\qquad
\langle\mathbf{f}(\mathbf{w}_B), \mathbf{q} \rangle=\frac{7}{12}q_1+\frac{1}{2}q_2,
\]
so we get
\[
\langle\mathbf{f}(\mathbf{w}_A), \hat{\mathbf{q}} \rangle-\langle\mathbf{f}(\mathbf{w}_B), \hat{\mathbf{q}} \rangle=\frac{1}{4}\hat{q}_2-\frac{10}{12}\hat{q}_1.
\]
Hence A is optimal iff $\hat{q}_2>\frac{10}{3}\hat{q}_1$, and B is optimal iff $\hat{q}_2<\frac{10}{3}\hat{q}_1$.

Now use a 13-round cycle of actual attacks:
\[
I_{13m+n}=
\begin{cases}
1, & n\in\{2,6,10\},\\
2, &\text{otherwise},
\end{cases}
\qquad m, n\in N, \ n\in[0,12]
\]

\begin{center}
\footnotesize
\setlength{\tabcolsep}{3pt}
\renewcommand{\arraystretch}{1.25}
\begin{tabular}{c|ccccccc}
$n$ & 1 & 2 & 3 & 4 & 5 & 6 & 7 \\
\hline
$l_t$ & 2 & 1 & 2 & 2 & 2 & 1 & 2 \\
$\hat q({t-1})$ & $(0,0)$ & $(0,1)$ & $(\frac12,\frac12)$ & $(\frac13,\frac23)$ & $(\frac14,\frac34)$ & $(\frac15,\frac45)$ & $(\frac13,\frac23)$ \\
$\hat q({t})$ & $(0,1)$ & $(\frac12,\frac12)$ & $(\frac13,\frac23)$ & $(\frac14,\frac34)$ & $(\frac15,\frac45)$ & $(\frac13,\frac23)$ & $(\frac27,\frac57)$ \\
$\mathbf{w}_{t-1}^*$ & $\sim$ & A & B & B & B & A & B \\
$\mathbf{w}_{t}^*$ & A & B & B & B & A & B & B
\end{tabular}

\begin{tabular}{c|cccccc}
$n$ & 8 & 9 & 10 & 11 & 12 & 13 \\
\hline
$l_t$ & 2 & 2 & 1 & 2 & 2 & 2 \\
$\hat q({t-1})$ & $(\frac27,\frac57)$ & $(\frac14,\frac34)$ & $(\frac29,\frac79)$ & $(\frac3{10},\frac7{10})$ & $(\frac3{11},\frac8{11})$ & $(\frac4{12},\frac8{12})$ \\
$\hat q({t})$ & $(\frac14,\frac34)$ & $(\frac29,\frac79)$ & $(\frac3{10},\frac7{10})$ & $(\frac3{11},\frac8{11})$ & $(\frac4{12},\frac8{12})$ & $(\frac3{13},\frac{10}{13})$ \\
$\mathbf{w}^*_{t-1}$ & B & B & A & B & B & B \\
$\mathbf{w}_{t}^*$ & B & A & B & B & B & $\sim$
\end{tabular}
\end{center}

The defender chooses strategy $\mathbf{w}_t$ in round $t$ according to
$( \hat q_1({t-1}),\hat q_2({t-1}))$.
For this period-$13$ construction,
$q_2-\frac{10}{3}q_1=0$ repeatedly oscillates around $0$. Strategy $\mathbf{w}_t^*$ is selected by optimizing $\langle\mathbf{f}(\mathbf{w}), \hat{\mathbf{q}}_{t} \rangle$. The strategy $\mathbf{w}_t=\mathbf{w}_{t-1}^*$ is selected by follow the leader algorithm (not follow the perturbed leader), $\mathbf{w}_t^*$ is selected by being the leader algorithm. Deterministic strategies $\mathbf{w}_t$ based on empirical frequencies (\(\hat{\mathbf{q}}\)) expose the defenders to exploitation by adaptive adversaries that are capable to infer protection patterns. Thus, adversaries may anticipate protection patterns and attack uncovered targets.
Hence, in every 13-round cycle, there exist five rounds in which the defender picks the wrong strategy between A and B.

Therefore, after $m$ full cycles ($T=13m$), cumulative regret satisfies $R(T)=\sum_{t=1}^T r_t\ge\frac{9}{4}m$.

We note that in Figures~\ref{Single attacker type with adversarial attack intensity sequence_Fig} and~\ref{Time-varying follower number with adversarial attack intensity sequence_Fig}, Baseline~2 performs better than Algorithm~1. This is because, in the single-attacker-type setting, Baseline~2 trivially converges to a fixed mixed strategy, which coincides with strategy~B. In contrast, our adversarial sequence is specifically constructed to mislead the follow-the-leader (FTL) algorithm by creating two strategies, A and B, whose cumulative payoffs remain nearly identical over most rounds. As a result, FTL is repeatedly induced to select the suboptimal strategy, whereas Baseline~2 consistently plays strategy~B. Moreover, during the $13m$ rounds, both A and B achieve the same payoff as the hindsight optimal strategy, and regret is incurred only during the remaining $n$ rounds where the two strategies differ. Therefore, the superior empirical performance of Baseline~2 in this setting is largely a consequence of this particular construction and is essentially trivial.

\section{Adversarial sequences with multiple attacker types}\label{Appendix12}
In this section, we illustrate a concrete example for the single-follower case shown in fig \ref{Multiple attacker type with adversarial attack intensity sequence_Fig}, where in each round there is only one attacker, but the attacker may be of different types and attack intensities. The adversary can design a sequence of attacker types and attack intensities over multiple rounds to systematically deceive the defender. By carefully constructing this sequence, the attacker ensures that the defender, who chooses strategies based only on historical empirical frequencies, repeatedly selects suboptimal strategies. This leads to a linear accumulation of regret over time, demonstrating how an adversarially designed sequence of attacker types and attack intensities can exploit the defender's reliance on past observations.

The adversary designs a repeated SSG to induce linear regret of the defender in Table \ref{Utility Matrix of Defender2}, \ref{Utility Matrix of Attacker2}.

\begin{table}[ht!]
  \begin{center}
    \caption{Defender's utility matrix}
    \label{Utility Matrix of Defender2}
    \begin{tabular}{@{}cccccc@{}}
      \toprule
      \textbf{\begin{tabular}[c]{@{}c@{}}Protected Target\end{tabular}} & \textbf{ Attack Target 1 } & \textbf{Attack Target 2} \\
      \midrule
      Target \ 1  & $\frac{1}{4}$ & -$\frac{1}{4}$ \\
      Target \ 2  & -$\frac{1}{4}$ & 1 \\
      \bottomrule
    \end{tabular}
  \end{center}
\end{table}

\begin{table}[ht!]
  \begin{center}
    \caption{Attacker $\alpha_1$'s utility matrix}
    \label{Utility Matrix of Attacker2}
    \begin{tabular}{@{}cccccc@{}}
      \toprule
      \textbf{\begin{tabular}[c]{@{}c@{}}Protected Target\end{tabular}} & \textbf{Attack Target 1} & \textbf{Attack Target 2} \\
      \midrule
       Target \ 1  & 0 & 1 \\
       Target \ 2  & \(\frac{1}{2}\) & 0 \\
      \bottomrule
    \end{tabular}
  \end{center}
\end{table}

\begin{table}[ht!]
  \begin{center}
    \caption{Attacker $\alpha_2$'s utility matrix}
    \label{Utility Matrix of Attacker3}
    \begin{tabular}{@{}cccccc@{}}
      \toprule
      \textbf{\begin{tabular}[c]{@{}c@{}}Protected Target\end{tabular}} & \textbf{Attack Target 1} & \textbf{Attack Target 2} \\
      \midrule
       Target \ 1  & 0 & $\frac{1}{2}$ \\
       Target \ 2  & \(1\) & 0 \\
      \bottomrule
    \end{tabular}
  \end{center}
\end{table}

We can partition the defender's strategy space into three regions according to the attackers' best responses:

\begin{itemize}
    \item \textbf{Region 1}: Both attacker $\alpha_1$ and $\alpha_2$ choose attack target 1 as their best response.
    \item \textbf{Region 2}: Attacker 1 chooses target 2 and attacker 2 chooses target 1.
    \item \textbf{Region 3}: Both attacker types choose target 2 as their best response .
\end{itemize}

Specifically, attacker 1's decision boundary is at $(\frac{1}{3},\frac{2}{3})$ and attacker 2's decision boundary is at $(\frac{2}{3},\frac{1}{3})$, which naturally divides the strategy space into three regions of distinct best-response combinations.

On each region, the defender's expected utility is a linear function of $\mathbf{w}$, and by the polytope/extreme-point argument, the candidate optimal strategies lie at the vertices of each region. Hence, the possible candidate strategies are:

\[
\begin{aligned}
\text{Region 1:} & \quad \mathbf{w}=(0,1), \quad \mathbf{w}=\left(\frac{1}{3},\frac{2}{3}\right),\\
\text{Region 2:} & \quad \mathbf{w}=\left(\frac{1}{3},\frac{2}{3}\right), \quad \mathbf{w}=\left(\frac{2}{3},\frac{1}{3}\right),\\
\text{Region 3:} & \quad \mathbf{w}=\left(\frac{2}{3},\frac{1}{3}\right), \quad \mathbf{w}=(1,0).
\end{aligned}
\]

By evaluating the vertices of each region (extreme points) with tie-breaking in favor of defender, we obtain the set of candidate strategies for the defender that may be optimal in each region.

The candidate strategies are:
\[
\text{A}:\ \mathbf{w}=(0,1)\ \text{(region 1)},
\]
\[
\text{B}:\ \mathbf{w}=\left(\frac{1}{3},\frac{2}{3}\right)\ \text{(region 2)}.
\]
\[
\text{C}:\ \mathbf{w}=(\frac{2}{3},\frac{1}{3})\ \text{(region 3)}.
\qquad
\
\]

Given distribution $\mathbf{q} = (q_{1,1},q_{1,2},q_{2,1},q_{2,2}) \in \Delta_4$, where $q_{kl}$ represents probability of attacker type $\alpha_k$ and attack intensity $l$.
The utility function is written as:

\[
\text{A}:\ \langle\mathbf{f}(\mathbf{w}_A), \mathbf{q} \rangle=-\frac{1}{4}q_{1,1}-\frac{1}{4}q_{2,1}+\frac{3}{4}q_{1,2}+\frac{3}{4}q_{2,2},\]
\[
\text{B}:\ \langle\mathbf{f}(\mathbf{w}_B), \mathbf{q} \rangle=\frac{7}{12}q_{1,1}-\frac{1}{12}q_{2,1}+\frac{1}{2}q_{1,2}+\frac{1}{2}q_{2,2},
\]
\[
\text{C}:\ \langle\mathbf{f}(\mathbf{w}_C), \mathbf{q} \rangle=\frac{1}{6}q_{1,1}+\frac{1}{6}q_{2,1}+\frac{1}{4}q_{1,2}+\frac{1}{4}q_{2,2},
\]

WLOG, we focus on two strategies $A$ and $B$, which are induced as optimal strategies for the attacker under an estimated sequence $\hat{\mathbf{q}}$. Here, $\hat{\mathbf{q}}$ is adversarially constructed so as to steer the attacker into selecting $A$ and $B$ as optimal responses. The equation is written as
\[
-\frac{1}{4}q_{1,1}-\frac{1}{4}q_{2,1}+\frac{3}{4}q_{1,2}+\frac{3}{4}q_{2,2}
=
\frac{7}{12}q_{1,1}-\frac{1}{12}q_{2,1}+\frac{1}{2}q_{1,2}+\frac{1}{2}q_{2,2},
\]
which characterizes the indifference condition between strategies $A$ and $B$.
It suffices that the empirical distribution constructed by the attacker oscillates across the two sides of this equality, so that $A$ and $B$ are alternately optimal. Then we have $H(\mathbf{q})=\frac{13}{3}q_{1,1}+\frac{5}{3}q_{2,1}-1$, where defender will choose strategy $A$ if $H(\mathbf{q})\leq0$ and choose strategy $B$ otherwise.

The foregoing outlines the decision-making criteria under the assumption that the attack intensity is known; however, since it is impossible to know its distribution in advance---and typically, one can only observe historical empirical frequencies---we will proceed to illustrate the process using these empirical frequencies.

\paragraph{Adversarial sequence} Defender at round $t$ can only observe attack intensities up to round $t-1$, and therefore uses at round $t$ the strategy that is optimal under $(\hat{q}_{1,1}({t-1}),\hat{q}_{1,2}({t-1}),\hat{q}_{2,1}({t-1}),\hat{q}_{2,2}({t-1}))$. At each round, the defender only observes history up to $t-1$, forms empirical frequencies
\[
\hat q_{k,l}({t-1})=\frac{1}{t-1}\sum_{s=1}^{t-1}\mathbb{I}\{\alpha_s=\alpha_k,l_s=l\},
\]
and at round $t$ plays the strategy optimal for $\hat{\mathbf{q}}(t-1)$.

Now use a 6-round cycle of actual attacks:
\[
I_{6m+n}=
\begin{cases}
\alpha=\alpha_1,l=1, & n=4,\\
\alpha=\alpha_2,l=1, &n=1,
\\
\alpha=\alpha_1\text{ or } \alpha_2,l=2, &\text{otherwise},
\end{cases}
\]
where $ m, n\in N, n\in[1,6]$.

\begin{center}
\footnotesize
\setlength{\tabcolsep}{3pt}
\renewcommand{\arraystretch}{1.25}
\begin{tabular}{c|cccccc}
$n$ & 1 & 2 & 3 & 4 & 5 & 6 \\
\hline
$\alpha_t,l_t$ & $\alpha_2$,1 & $\alpha_2$,2 & $\alpha_2$,2 & $\alpha_1$,1 & $\alpha_2$,2 & $\alpha_1$,2 \\
$m \cdot\mathbf{H}(\hat{\mathbf{ q}}({t-1}))$& 0 & $\frac{2}{3}$ & $-\frac16$ & $-\frac49$ & $\frac1{2}$ & $\frac1{5}$  \\
$m \cdot \mathbf{H}(\hat{\mathbf{ q}}({t}))$ & $\frac{2}{3}$ & $-\frac16$ & $-\frac49$ & $\frac1{2}$ & $\frac1{5}$ & 0 \\
$\mathbf{w}^*_{t-1}$ & $\sim$ & B & A & A & B & B \\
$\mathbf{w}^*_{t}$ & B & A & A & B & B & $\sim$
\end{tabular}
\end{center}

The defender chooses strategy $\mathbf{w}_t$ in round $t$ according to
$ \hat {\mathbf{q}}({t-1}))$.
For this construction, $H(\hat{\mathbf{q}})$ repeatedly oscillates around $0$. Strategy $\mathbf{w}_t^*$ is selected by optimizing $\langle\mathbf{f}(\mathbf{w}), \hat{\mathbf{q}}_{t} \rangle$. Strategy $\mathbf{w}_t=\mathbf{w}_{t-1}^*$ is selected by follow the leader algorithm (not follow the perturbed leader), $\mathbf{w}_t^*$ is selected by being the leader algorithm. Deterministic strategies $\mathbf{w}_t=\mathbf{w}_{t-1}^*$ based on empirical frequencies (\(\hat{\mathbf{q}}\)) expose the defenders to exploitation by adaptive adversaries that are capable to infer protection patterns. Thus, adversaries may anticipate protection patterns and attack uncovered targets.

Hence, in every 6-round cycle, there exist two rounds in which the defender picks the wrong strategy between A and B. Therefore, after $m$ full cycles ($T=6m$), cumulative regret satisfies $R(T)=\sum_{t=1}^T r_t\ge \frac{7}{12}m$.

\section{Adversarial sequences with multiple followers and types}\label{Appendix_Simulation22}

In this section, we consider a multiple-follower setting shown in fig \ref{Simulation 1}, where a fixed number of attackers appear in each round. While the number of followers remains constant, their types and attack intensities may vary across rounds. We construct an adversarial sequence over attacker types and attack intensities to systematically mislead the defender. By carefully coordinating the joint behavior of multiple followers, the adversary ensures that the empirical distribution observed by the defender induces alternating optimal strategies, causing the defender to repeatedly select suboptimal actions. As a result, the defender incurs linear regret over time.

Consider the multiple follower scenario (two follower). The utility matrices are referred to Table \ref{Utility Matrix of Defender2}, \ref{Utility Matrix of Attacker2}, \ref{Utility Matrix of Attacker3}.
Given vector $\mathbf{q} = (q_{1,1},q_{1,2},q_{2,1},q_{2,2})$, where $q_{kl}$ represents frequency of attacker type $\alpha_k$ and attack intensity $l$ and $\sum_{k,l} q_{kl}=2$.

By evaluating the vertices of each region (extreme points) with tie-breaking in favor of defender, we obtain the set of candidate strategies for the defender that may be optimal in each region.
The candidate strategies are:
\[
\text{A}:\ \mathbf{w}=(0,1)\ \text{(region 1)},
\]
\[
\text{B}:\ \mathbf{w}=\left(\frac{1}{3},\frac{2}{3}\right)\ \text{(region 2)}.
\]
\[
\text{C}:\ \mathbf{w}=(\frac{2}{3},\frac{1}{3})\ \text{(region 3)}.
\qquad
\
\]

Given distribution $\mathbf{q} = (q_{1,1},q_{1,2},q_{2,1},q_{2,2})$, with $\sum_{k,l} q_{kl}=2$.
The utility function is written as:

\[
\text{A}:\ \langle\mathbf{f}(\mathbf{w}_A), \mathbf{q} \rangle=-\frac{1}{4}q_{1,1}-\frac{1}{4}q_{2,1}+\frac{3}{4}q_{1,2}+\frac{3}{4}q_{2,2},\]
\[
\text{B}:\ \langle\mathbf{f}(\mathbf{w}_B), \mathbf{q} \rangle=\frac{7}{12}q_{1,1}-\frac{1}{12}q_{2,1}+\frac{1}{2}q_{1,2}+\frac{1}{2}q_{2,2},
\]
\[
\text{C}:\ \langle\mathbf{f}(\mathbf{w}_C), \mathbf{q} \rangle=\frac{1}{6}q_{1,1}+\frac{1}{6}q_{2,1}+\frac{1}{4}q_{1,2}+\frac{1}{4}q_{2,2},
\]

We focus on two strategies $A$ and $B$, which are induced as optimal strategies for the attacker under an estimated sequence $\hat{\mathbf{q}}$. Here, $\hat{\mathbf{q}}$ is adversarially constructed so as to steer the attacker into selecting $A$ and $B$ as optimal responses. The equation is written as
\[
-\frac{1}{4}q_{1,1}-\frac{1}{4}q_{2,1}+\frac{3}{4}q_{1,2}+\frac{3}{4}q_{2,2}
=
\frac{7}{12}q_{1,1}-\frac{1}{12}q_{2,1}+\frac{1}{2}q_{1,2}+\frac{1}{2}q_{2,2},
\]
which characterizes the indifference condition between strategies $A$ and $B$.
It suffices that the empirical distribution constructed by the attacker oscillates across the two sides of this equality, so that $A$ and $B$ are alternately optimal. Then we have $H(\mathbf{q})=\frac{13}{3}q_{1,1}+\frac{5}{3}q_{2,1}-2$, where defender will choose strategy $A$ if $H(\mathbf{q})\leq0$ and choose strategy $B$ otherwise.

The foregoing outlines the decision-making criteria under the assumption that the attack intensity is known; however, since it is impossible to know its distribution in advance---and typically, one can only observe historical empirical frequencies---we will proceed to illustrate the process using these empirical frequencies.

\paragraph{Adversarial sequence} Defender at round $t$ can only observe attack intensities up to round $t-1$, and therefore uses at round $t$ the strategy that is optimal under $(\hat{q}_{1,1}({t-1}),\hat{q}_{1,2}({t-1}),\hat{q}_{2,1}({t-1}),\hat{q}_{2,2}({t-1}))$. At each round, the defender only observes history up to $t-1$, forms empirical frequencies
\[
\hat q_{k,l}({t-1})=\frac{1}{t-1}\sum_{s=1}^{t-1}\sum_{c=1}^2\mathbb{I}\{\alpha_{s,c}=\alpha_k,l_{s,c}=l\},
\]
,where $c$ denotes the follower $c$. At round $t$ plays the strategy optimal for $\hat{\mathbf{q}}(t-1)$.

Now use a 3-round cycle of actual attacks:

where $ m, n\in N, n\in[1,3]$.

\begin{center}
\footnotesize
\setlength{\tabcolsep}{3pt}
\renewcommand{\arraystretch}{1.25}
\begin{tabular}{c|ccc}
$n$ & 1 & 2 & 3  \\
\hline
follower 1 & $\alpha_2$,2 & $\alpha_1$,1 & $\alpha_2$,2  \\
follower 2 & $\alpha_1$,2 & $\alpha_2$,2 & $\alpha_2$,1  \\
$m \cdot\mathbf{H}(\hat{\mathbf{ q}}({t-1}))$& 0 & -2 & $\frac{1}{6}$  \\
$m \cdot \mathbf{H}(\hat{\mathbf{ q}}({t}))$ & -2 & $\frac{1}{6}$ & 0  \\
$\mathbf{w}_{t-1}^*$ & $\sim$ & A & B  \\
$\mathbf{w}_{t}^*$ & A & B & $\sim$
\end{tabular}
\end{center}

The defender chooses strategy $\mathbf{w}_t$ in round $t$ according to
$ \hat {\mathbf{q}}({t-1}))$.
For this construction, $H(\hat{\mathbf{q}})$ repeatedly oscillates around $0$. Strategy $\mathbf{w}_t^*$ is selected by optimizing $\langle\mathbf{f}(\mathbf{w}), \hat{\mathbf{q}}_{t} \rangle$. Strategy $\mathbf{w}_t=\mathbf{w}_{t-1}^*$ is selected by follow the leader algorithm (not follow the perturbed leader), $\mathbf{w}_t^*$ is selected by being the leader algorithm. Deterministic strategies $\mathbf{w}_t=\mathbf{w}_{t-1}^*$ based on empirical frequencies (\(\hat{\mathbf{q}}\)) expose the defenders to exploitation by adaptive adversaries that are capable to infer protection patterns. Thus, adversaries may anticipate protection patterns and attack uncovered targets.

Hence, in every 3-round cycle, there exist one rounds in which the defender picks the wrong strategy between A and B. Therefore, after $m$ full cycles ($T=3m$), cumulative regret satisfies $R(T)=\sum_{t=1}^T r_t\ge \frac{m}{12}$.

\section{Experiments with a time-varying follower population}\label{Appendix_Simulation3}

In this section, we consider a setting with a time-varying number of followers under full information feedback shown in fig \ref{Time-varying follower number with adversarial attack intensity sequence_Fig}. In each round, multiple attackers may appear, but the number of active followers can vary (e.g., no-shows). We construct an adversarial sequence over both the number of active followers and their attack intensities to systematically mislead the defender. By carefully designing this sequence, the adversary ensures that the empirical frequency estimates based on historical observations repeatedly induce incorrect strategy choices. As a result, the defender is forced to alternate between suboptimal decisions and incurs linear regret over time. Additionally, in partial information feedback scenario, only aggregate attack effect are observable without identifying individual attack intensities or follower numbers.

Consider the multiple follower scenario (two follower). The utility matrices are referred to Table \ref{Utility Matrix of Defender}, \ref{Utility Matrix of Attacker}.
Given vector $\mathbf{q} = (q_{1,1},q_{1,2})$, where $q_{l}$ represents frequency of attacker's with attack intensity $l$.

By evaluating the vertices of each region (extreme points) with tie-breaking in favor of defender, we obtain the set of candidate strategies for the defender that may be optimal in each region:
\[
\text{A}:\ \mathbf{w}=(0,1)\ \text{(region 1)},
\qquad
\text{B}:\ \mathbf{w}=\left(\frac{1}{3},\frac{2}{3}\right)\ \text{(region 2)}.
\]

By comparing A and B, we have
\[
-q_1+3q_2>\frac{7}{3}q_1+2q_2
\iff q_2>\frac{10}{3}q_1,
\]
which means A is better than B when $q_2>\frac{10}{3}q_1$; otherwise, $q_2<\frac{10}{3}q_1$,
B is better than A.

The above formulation assumes that the attack intensity frequency is known in advance. In practice, however, its frequency is unknown and only empirical frequencies from historical data are available.  We therefore base our analysis on these empirical estimates.

\textbf{Adversarial sequence for $q_1,q_2$}

Assume that the defender in round $t$ can only observe attack intensities up to round $t-1$. Accordingly, it selects a strategy that is optimal with respect to the empirical estimates $(\hat{q}_1(t-1), \hat{q}_2(t-1))$. At each round $t$, the attacker selects an intensity $l \in \{1,2\}$.  The defender observes only the history up to round $t-1$, forms empirical frequency estimates accordingly
\[
\hat q_1({t-1})=\frac{1}{t-1}\sum_{s=1}^{t-1}\sum_{c=1}^2\mathbb{I}\{l_{s,c}=1\},
\qquad
\hat q_2({t-1})=\frac{1}{t-1}\sum_{s=1}^{t-1}\sum_{c=1}^2\mathbb{I}\{l_{s,c}=2\},
\]
and plays the strategy optimal for
$( \hat q_1({t-1}),\hat q_2({t-1}))$.

Now use a 7-round cycle of actual attacker sequences:

\begin{center}
\footnotesize
\setlength{\tabcolsep}{3pt}
\renewcommand{\arraystretch}{1.25}
\begin{tabular}{c|ccccccc}
$n$ & 1 & 2 & 3 & 4 & 5 & 6 & 7 \\
\hline
follower 1 & l=2 & l=1 & l=2 & l=1 & l=2 & l=1 & l=2 \\
follower 2 & l=2 & no show & l=2 & l=2 & l=2 & l=2 & l=2 \\
$\mathbf{w}_{t-1}$ & $\sim$ & A & B & A & B & A & B \\
$\mathbf{w}_{t}$ & A & B & A & B & A & B & $\sim$
\end{tabular}
\end{center}

The defender chooses strategy $\mathbf{w}_t$ in round $t$ according to
$( \hat q_1({t-1}),\hat q_2({t-1}))$.
For this period-$7$ construction,
$q_2-\frac{10}{3}q_1=0$ repeatedly oscillates around $0$. Strategy $\mathbf{w}_t^*$ is selected by optimizing $\langle\mathbf{f}(\mathbf{w}), \hat{\mathbf{q}}_{t} \rangle$. The strategy $\mathbf{w}_t=\mathbf{w}_{t-1}^*$ is selected by follow the leader algorithm (not follow the perturbed leader), $\mathbf{w}_t^*$ is selected by being the leader algorithm. Deterministic strategies $\mathbf{w}_t$ based on empirical frequencies (\(\hat{\mathbf{q}}\)) expose the defenders to exploitation by adaptive adversaries that are capable to infer protection patterns. Thus, adversaries may anticipate protection patterns and attack uncovered targets.
Hence, in every 7-round cycle, there exist five rounds in which the defender picks the wrong strategy between A and B.

Therefore, after $m$ full cycles ($T=7m$), cumulative regret satisfies $R(T)=\sum_{t=1}^T r_t\ge \frac{3}{2}m$.

\begin{figure}[hpbt]
\centering
\begin{subfigure}[t]{0.45\linewidth}
    \centering
    \includegraphics[width=\linewidth]{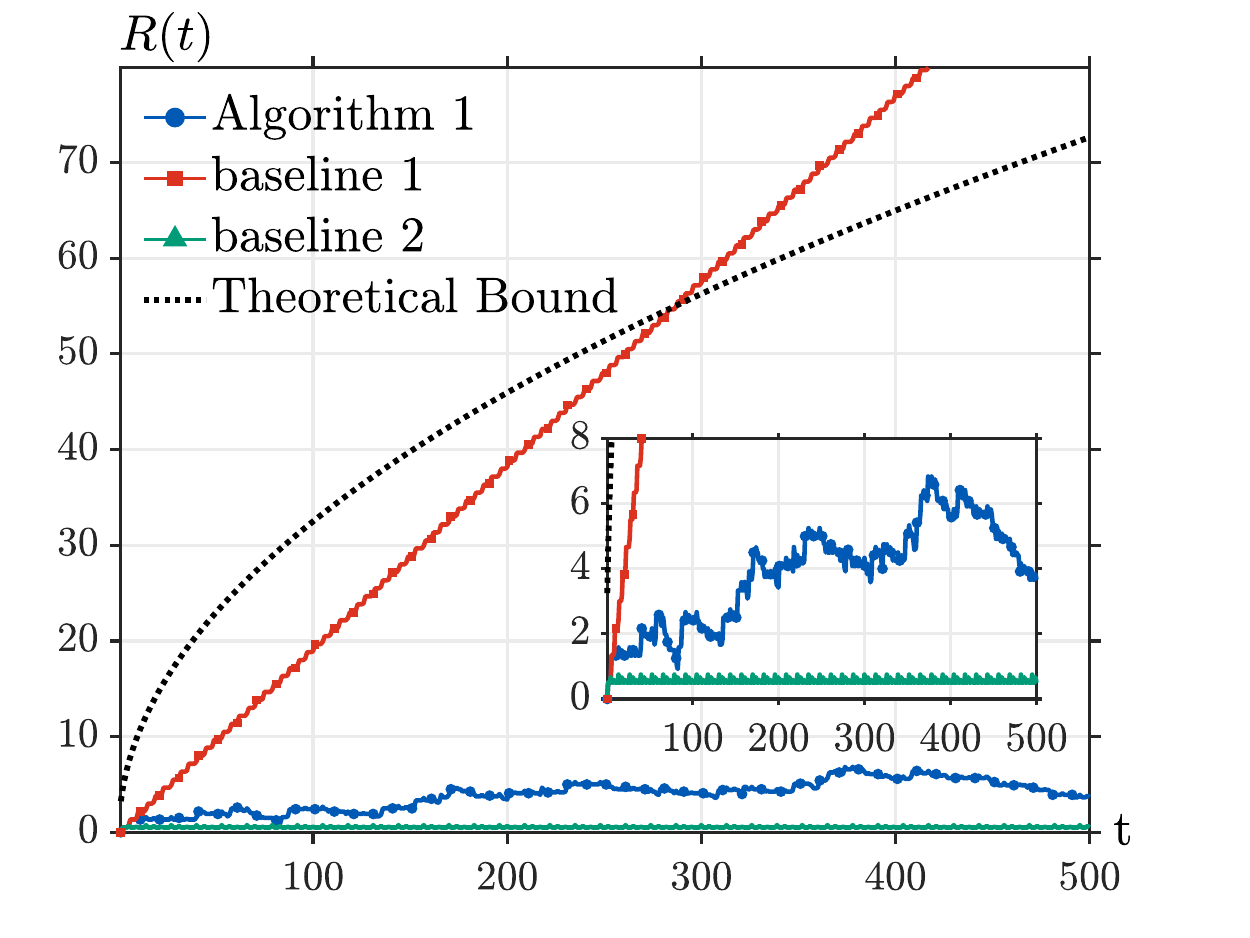}
    \caption{Single attacker type with adversarial attack intensity sequence}
    \label{Single attacker type with adversarial attack intensity sequence_Fig}
\end{subfigure}
\hfill
\begin{subfigure}[t]{0.45\linewidth}
    \centering
    \includegraphics[width=\linewidth]{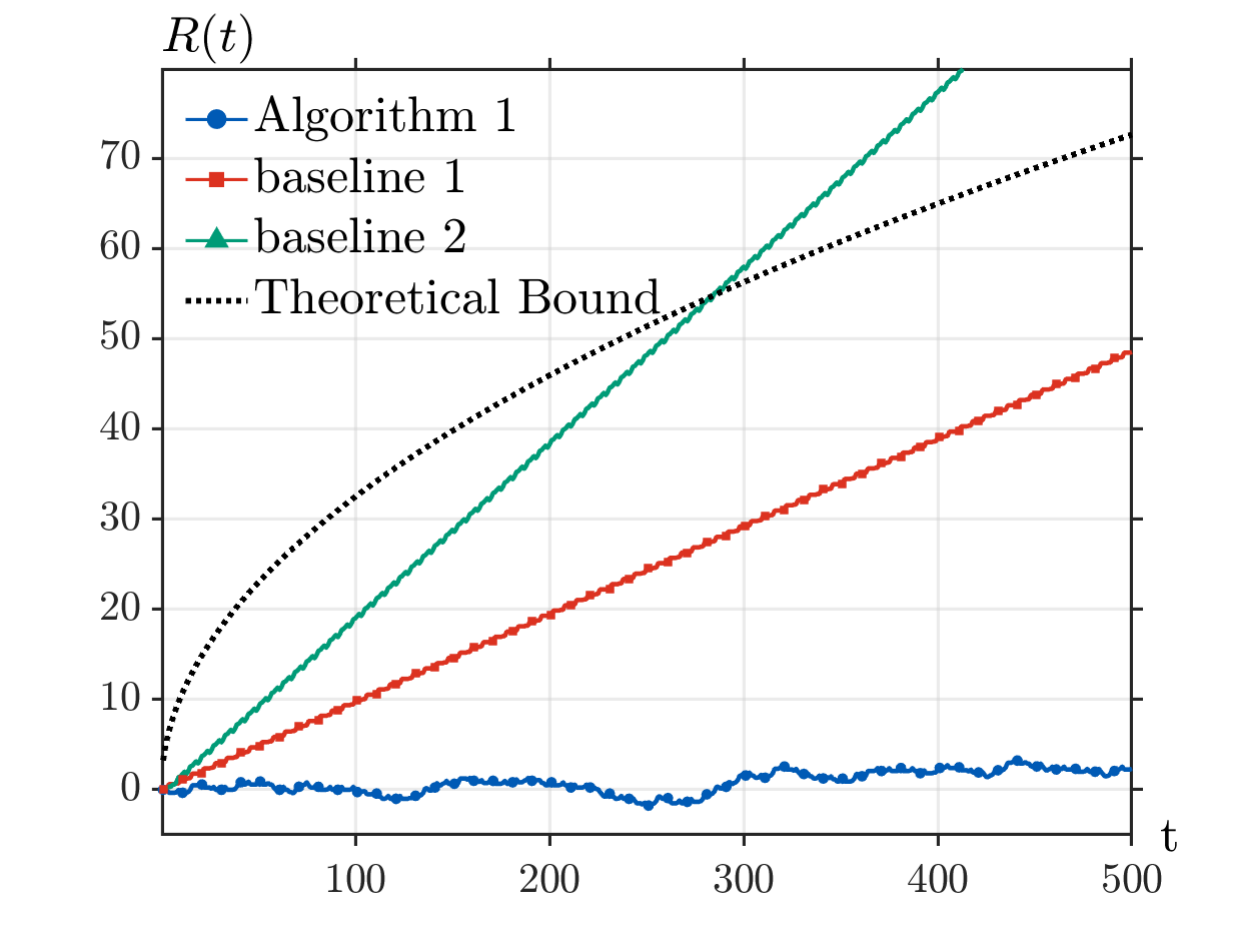}
    \caption{Multiple attacker type with adversarial attack intensity sequence}
    \label{Multiple attacker type with adversarial attack intensity sequence_Fig}
\end{subfigure}
\caption{Regret comparison under different adversarial settings.}
\label{fig:regret_single_attacker}
\end{figure}

\begin{figure}[hpbt]
\centering
\begin{subfigure}[t]{0.45\linewidth}
    \centering
    \includegraphics[width=\linewidth]{figures/ICLRfigure4right.pdf}
    \caption{barycentric-spanner EXP2}
    \label{Sum_Regret_Partial_Info Feedback_with_adversarial_Attacker_Sequence_Fig}
\end{subfigure}
\hfill
\begin{subfigure}[t]{0.45\linewidth}
    \centering
    \includegraphics[width=\linewidth]{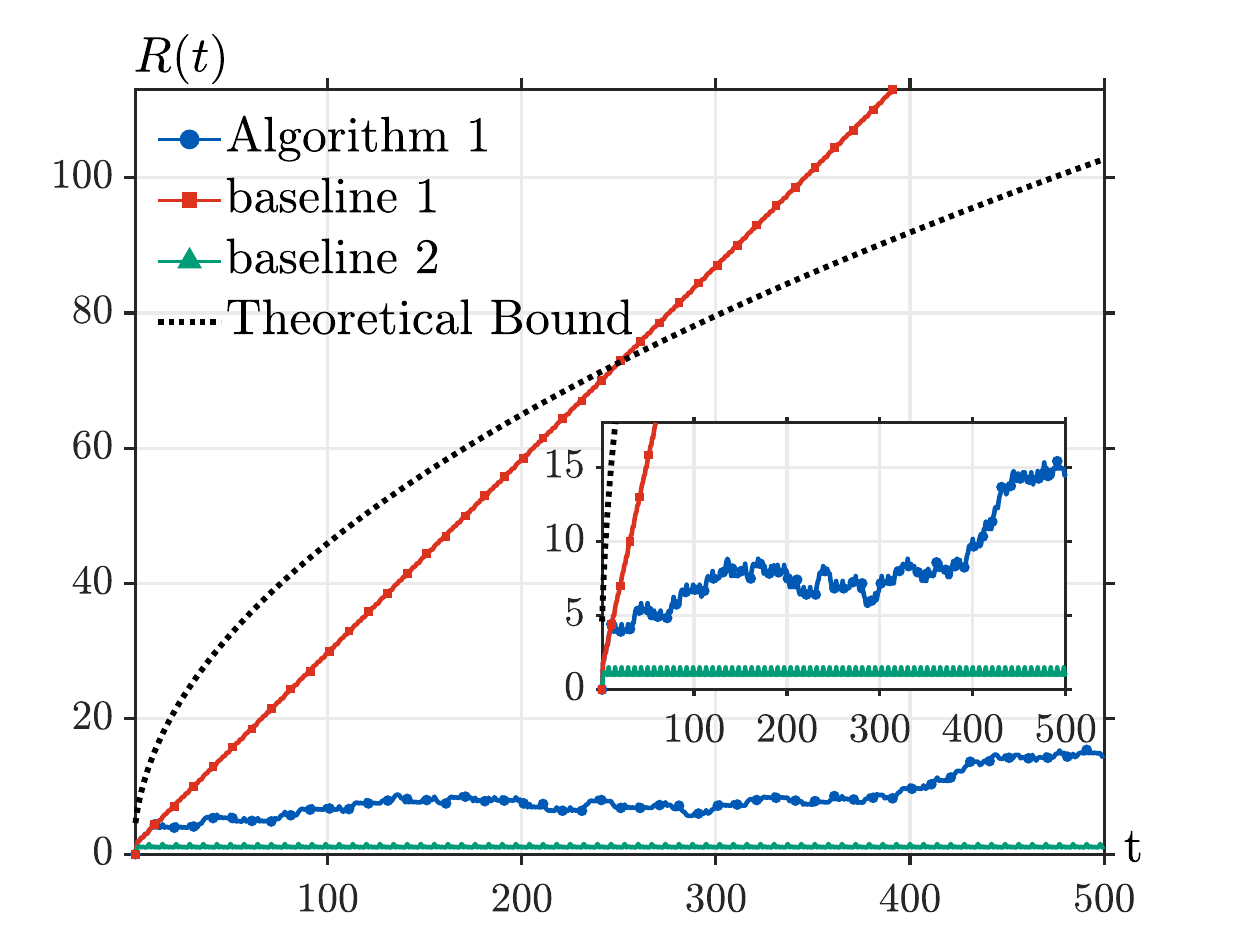}

    \caption{Time-varying follower number with adversarial attack intensity sequence}
    \label{Time-varying follower number with adversarial attack intensity sequence_Fig}
\end{subfigure}

    \caption{Retained results: (a) barycentric-spanner EXP2; (b) full-information FPL with a time-varying follower population.}
    \label{Sum Regret Partial 1}
\end{figure}

\section{Aggregate-feedback learning with two followers}
\label{Appendix_Simulation4}

We consider two targets, two attacker types, and exactly two
followers per round, so that $N=K=F=C=2$.
The defender chooses $\mathbf w\in[0,1]^2$ with $w_1+w_2=1$.
Each follower attacks one or two targets, and best-response
ties favor the defender. The utility parameters are
\[
    \boldsymbol{\lambda}^{r}=(-1,-1),
    \qquad
    \boldsymbol{\rho}^{r}=(1,1),
    \qquad
    U_a^u=
    \begin{bmatrix}
        1/4&1\\
        1&1/4
    \end{bmatrix},
    \qquad
    U_a^c=
    \begin{bmatrix}
        0&0\\
        0&0
    \end{bmatrix}.
\]
Rows index targets and columns index attacker types.
The candidate set is
\[
    V=\{(0,1),(1/5,4/5),(4/5,1/5),(1,0)\}.
\]
The intensity-two payoffs are zero for every strategy, and the
payoff vectors have effective dimension $d=2$.
We use the first two candidates as the basis strategies.

\paragraph{Experimental setting.}
Every 250 rounds, the type probabilities alternate between
$(0.85,0.15)$ and $(0.75,0.25)$, while the intensity
probabilities alternate between $(0.9,0.1)$ and $(0.8,0.2)$.
Types and intensities are sampled independently for each
follower and round, conditional on the regime.
The complete attacker sequence is generated before learning.
Both regimes have the same unique expected-optimal strategy,
$(1/5,4/5)$, with payoff gaps of $1.512$ and $0.960$ to the
second-best candidate. This stable optimum makes the instance
suitable for illustrating how the algorithm learns from
aggregate feedback.

Since $\log|V|/T$ can be small, we use a larger constant
multiplier in the learning rate while retaining its
$T^{-1/2}$ scaling and enforcing $\gamma=d^2\eta\le1/2$,
which changes only the constant in the
$\mathcal{O}(\sqrt{T})$ regret bound.
\paragraph{Parameter choice.}
We use the original EXP2 update and normalization
$\ell_t(\mathbf w)=-r_t(\mathbf w)/(CF)$.
For each total horizon $T$, choose
\begin{equation}
    \eta=
    \min\left\{
        \frac{1}{2d^2},
        \frac{2}{d}\sqrt{\frac{\log|V|}{T}}
    \right\},
    \qquad
    \gamma=d^2\eta.
    \label{eq:example-fixed-parameters}
\end{equation}
Both parameters remain fixed throughout each run.
The learning-rate coefficient was selected in a separate
preliminary comparison of one, two, and four times the
original recommended coefficient.

The general regret bound allows any parameters satisfying
$\eta d^2\le\gamma$ and $0<\gamma<1$.
Our choice satisfies
$\eta d^2=\gamma\le1/2$, so the original proof applies:
\begin{align*}
    \mathbb E[R(T)]
    &\le CF\left[
        \frac{\log|V|}{\eta}
        +\frac{\eta dT}{1-\gamma}
        +2\gamma T
    \right]\\
    &\le CF\left[
        \frac{\log|V|}{\eta}
        +4\eta d^2T
    \right].
\end{align*}
For $T\ge16d^2\log|V|$, substitution gives
\begin{equation}
    \mathbb E[R(T)]
    \le\frac{17}{2}CFd\sqrt{T\log|V|}
    =68\sqrt{T\log4}.
    \label{eq:example-fixed-bound}
\end{equation}
For smaller $T$, the same bound follows from
$R(T)\le2CF T\le8CFd\sqrt{T\log|V|}$.
Thus, the larger learning rate preserves the square-root
dependence on $T$, although it increases the bound's constant.

\paragraph{Results.}
We use
$T\in\{1000,2500,5000,10000,20000\}$,
with ten independently generated attacker sequences and
32 algorithm repetitions per sequence.
For each horizon, the algorithm starts from uniform weights
and uses the fixed parameters in
\eqref{eq:example-fixed-parameters}.
Regret is evaluated against the best fixed candidate in
hindsight; only the selected strategy's aggregate payoff
is used for learning.

For $T=10{,}000$, the parameters are
$\eta=0.01177410$ and $\gamma=0.04709640$.
The mean probability of selecting $(1/5,4/5)$ reaches
$0.9489$ at round 1,000 and $0.9763$ at round 2,500.
The final mean cumulative regret is $999.59$, with standard
error $2.52$ across the ten sequence-level means.

\begin{figure}[t]
    \centering
    \includegraphics[width=0.65\linewidth]{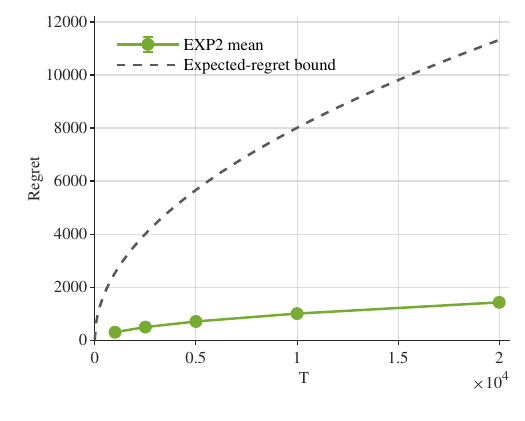}
    \caption{
    Aggregate-feedback EXP2 with fixed parameters..
    }
    \label{fig:exp2-two-followers}
\end{figure}
